\documentclass[11pt]{article}
\usepackage{graphics,color}
\usepackage{graphicx}
\usepackage{amssymb,amsmath,amsfonts}
\usepackage{amsmath,mathrsfs,bm,url,times}
\usepackage{latexsym}
\usepackage{subfigure}
\usepackage{algorithmic}
\usepackage{algorithm}
\usepackage{epsfig,url}

\newtheorem{theorem}{Theorem}
\newtheorem{lemma}{Lemma}
\newtheorem{definition}{Definition}

\newtheorem{remark}{Remark}
\newtheorem{corollary}{Corollary}
\newtheorem{property}{Property}
\newenvironment{proof}{\vspace{1ex}\noindent{\bf Proof.}\hspace{0.5em}}
    {\hfill\qed\vspace{1ex}}
\def\qed{\hfill \vrule height 6pt width 6pt depth 0pt}

\newcommand{\R}{{\mathbb R}}
\newcommand{\onetom}{1,\cdots,m}
\newcommand{\oneton}{1,\cdots,n}
\newcommand{\onetoK}{1,\cdots,K}

\title{Event-triggered Consensus for Multi-agent Systems with Asymmetric and Reducible Topologies}
\author{Xinlei Yi, Wenlian Lu and Tianping Chen}

\begin{document}
\maketitle
\begin{abstract}
This paper studies the consensus problem of multi-agent systems with
asymmetric and reducible topologies. Centralized event-triggered rules are
provided so as to reduce the frequency of system's updating. The diffusion
coupling feedbacks of each agent are based on the latest observations from
its in-neighbors and the system's next observation time is triggered by a
criterion based on all agents' information. The scenario of continuous
monitoring is first considered, namely all agents' instantaneous states can
be observed. It is proved that if the network topology has a spanning tree,
then the centralized event-triggered coupling strategy can realize
consensus for the multi-agent system. Then the results are extended to
discontinuous monitoring, where the system computes its next triggering
time in advance without having to observe all agents' states continuously.
Examples with numerical simulation are provided to show the effectiveness
of the theoretical results.

\end{abstract}
\section{Introduction}
In the past decade, consensus problem in multi-agent systems, i.e. a group
of agents seeks to agree upon certain quantity of interest, has attracted
many researchers. There are many excellent results in this field, for
example, see \cite{Ros}-\cite{Liubo}. In these works, the network
topologies can be fixed or stochastically switching, and to realize a
consensus a fundamental assumption is that the underlying graph of the
network system has a spanning tree \cite{Ros}.

However, the above studies are all using the simultaneous state as feedback
control to realize a consensus. In the near future, each agent could be
equipped with embedded microprocessors with limited resources that will
transmit and gather information, etc. Motivated by that, event-triggered
control \cite{Pta}-\cite{Yfg} and self-triggered control
\cite{Aapt}-\cite{Aap} have been proposed and studied. Instead of using the
simultaneous state to realize a consensus, the control in event-triggered
control strategy is piecewise constant between the triggering times which
need been determined. Self-triggered control is a natural extension of the
event-triggered control since the derivative of the concern multi-agent
system's state is piecewise constant (a very simple linear constant
coefficient ordinary differential equations) in mathematical respect, which
means it is very easy to work out solutions (agents' states) of the
equations. In \cite{Pta}, the triggering times are determined when a
certain error becomes large enough with respect to the norm of the state.
In \cite{Dvd}, under the condition that the graph is undirected and
strongly connected, the authors provide event-triggered and self-triggered
approaches in both centralized and distributed formulations. It should be
emphasized that the approaches cannot be applied to directed graphes. In
\cite{Zliu}, the authors investigate the average-consensus problem of
multi-agent systems with directed and weighted topologies, but they need an
additional assumption that the directed topology must be balanced. In
\cite{Yfg}, the authors propose a new combinational measurement approach to
event design and as a result, control of agents is only triggered at their
own even time, which is an improvement.

In fact, event-driven strategies for multi-agent systems can be viewed as
linearization and discretization process, which has been considered and
investigated in early papers \cite{LC2004,LC2007}. For example, in the
paper \cite{LC2004}, following algorithm was investigated
\begin{align}
x^{i}(t+1)=f(x^{i}(t))+c_{i}\sum_{j=1}^{m}a_{ij}(f(x^{j}(t)))
\end{align}
which can be considered as nonlinear consensus algorithm.

As a special case, let $f(x(t))=x(t)$ and $c_{i}=(t_{k+1}^{i}-t_{k}^{i})$,
then
\begin{align}
x^{i}(t_{k+1}^{i})=x^{i}(t_{k}^{i})+(t_{k+1}^{i}-t_{k}^{i})\sum_{j=1}^{m}a_{ij}x^{j}(t_{k}^{i})
\end{align}
which is just the event triggering (distributed, self triggered) model for
consensus problem, though the term "event triggering" is not used. In
centralized control, the bound for
$(t_{k+1}^{i}-t_{k}^{i})=(t_{k+1}-t_{k})$ to reach synchronization was
given in that paper when the coupling graph is indirected, too.

In this paper, continuing with previous works, we study centralized
event-triggered and centralized self-triggered consensus in multi-agent
system with asymmetric, reducible and weighted topology. Firstly, under the
irreducible topology condition, i.e. the underlying directed graph is
strongly connected, we derive two centralized event-triggered rules.
Secondly, by mathematical induction, we generalize above results to
reducible topology case. It is proved that if the network topology has a
spanning tree, then the centralized event-triggered coupling strategies can
realize consensus for the multi-agent system. Here, we point out that we do
not need that the directed topology must be balanced or any other
additional conditions. The consensus value is a weighted average of all
agents' initial values by the nonnegative left eigenvector of the graph
Laplacian matrix corresponding to eigenvalue zero. Finally, the results are
extended to discontinuous monitoring, i.e. self-triggered control, where
the system computes its next triggering time in advance without having to
observe all agents' states continuously. In addition, we provide a novel
and very simple self-triggered rule (see Theorem \ref{thm0.12}) of which
the time interval length is bigger than the result in \cite{Dvd}. And we
even give a self-triggered strategy with a fixed time interval between two
continuous triggering, i.e. we give a periodic self-triggered strategy. It
should be emphasized that the period is decided by the maximum in-degree of
the network only.

The paper is organized as follows: in Section 2, some necessary definitions
and lemmas are given; in Section 3, the event-triggered consensus for
multi-agent with asymmetrical topologies is discussed; in Section 4, the
self-triggered formulation of the frameworks in Section 3 is presented; in
Section 5, examples with numerical simulation are provided to show the
effectiveness of the theoretical results; in Section 6, we conclude this
paper and indicate further research directions.

\section{Preliminaries}
In this section we first review some related definitions and results on
algebraic graph theory \cite{Die,Rah} which will be used later in this
paper.

For a weighted asymmetric graph (digraph or directed graph) $\mathcal
G=(\mathcal V, \mathcal E, \mathcal A)$ with $m$ agents (vertices or
nodes), the set of agents $\mathcal V =\{v_1,\cdots,v_m\}$, set of links
(edges) $\mathcal E \subseteq \mathcal V \times \mathcal V$, and the
weighted adjacency matrix $\mathcal A =(a_{ij})$ with nonnegative adjacency
elements $a_{ij}$. A link of $\mathcal G$ is denoted by $e(i,j)=(v_i,
v_j)\in \mathcal E$ if there is a directed link from agent $j$ to agent
$i$. The adjacency elements associated with the links of the graph are
positive, i.e., $e(i,j)\in \mathcal E\iff a_{ij}>0$, for all $i,
j\in\mathcal I$, where $\mathcal I=\{1, 2,\cdots, m\}$. It is assumed that
$a_{ii}=0$ for all $i\in \mathcal I$. We define the link set $\mathcal
E^{in}=\{e^{in}(i,j)\}$ is composed of directed links $e^{in}(i,j)$ if
$a_{ij}>0$, i.e., $\mathcal E^{in}=\mathcal E$. Equivalently, we can define
a dual link set $\mathcal E^{out}=\{(e^{out}(i,j)\}$ composed of
$e^{out}(i,j)$ if $a_{ji}>0$. Moreover, the in-neighbours set of agent
$v_i$ is defined as
\begin{eqnarray*}
N^{in}_i=\{v_j\in \mathcal V\mid a_{ij}>0\}.
\end{eqnarray*}
The in-degree of agent $v_i$ is defined as follows:
\begin{eqnarray*}
deg^{in}(v_i)=\sum\limits_{j=1}^{m}a_{ij}.
\end{eqnarray*}

The degree matrix of digraph $\mathcal G$ is defined as
$D=diag[deg^{in}(v_1), \cdots, deg^{in}(v_m)]$. The weighted Laplacian
matrix associated with the digraph $\mathcal G$ is defined as $L=\mathcal
A-D$. A directed path from agent $v_0$ to agent $v_k$ is a directed graph
with distinct agents $v_0,...,v_k$ and links $e_0,...,e_{k-1}$ such that
$e_i$ is a link directed from $v_i$ to $v_{i+1}$, for all $i<k$.
\begin{definition}
We say an asymmetric graph $\mathcal G$ is strongly connected if for any
two distinct agents $v_{i},~v_{j}$, there exits a directed path from
$v_{i}$ to $v_{j}$.
\end{definition}
By \cite{Rah}, we know that $\mathcal G$ is strongly connected is
equivalent to the corresponding Laplacian matrix $L$ is irreducible.
\begin{definition}
We say an asymmetric graph $\mathcal G$ has a spanning tree if exists at
least one agent $v_{i_{0}}$ such that for any other agent $v_{j}$, there
exits a directed path from $v_{i_{0}}$ to $v_{j}$.
\end{definition}

By Perron-Frobenius theorem \cite{Rah}, we have
\begin{lemma}\label{lem2}
If $L$ is irreducible, then $rank(L)=m-1$, zero is an algebraically simple
eigenvalue of $L$ and there is a positive vector
$\xi^{\top}=[\xi_{1},\cdots,\xi_{m}]$ such that $\xi^{\top} L=0$ and
$\sum_{i=1}^{m}\xi_{i}=1$. And if the asymmetric graph $\mathcal G$ just
has a spanning tree then we should change the positive vector to
nonnegative vector in above conclusion.
\end{lemma}

Let $\Xi=diag[\xi_{1},\cdots,\xi_{m}]$, also by Perron-Frobenius theorem, we have
\begin{lemma}\label{lem1}
If $L$ is irreducible, then $\Xi L+L^{\top}\Xi$ is a symmetric Metzler
matrix with all row sums equal to zeros and has zero eigenvalue with
algebraic dimension one.
\end{lemma}

Denote $R=[R_{ij}]_{i,j=1}^{m}=(1/2)(\Xi L+L^{\top}\Xi)$. By Lemma
\ref{lem1}, $R$ can be regarded as a Laplacian matrix of some bi-directed
graph with strongly connected topology. Denote this graph as $\mathcal
G^{s}=\{\mathcal V,\mathcal E^{s}\}$ with the same agent set as $\mathcal
G$, and the link set $\mathcal E^{s}$ is composed of $e^{s}(i,j)$ for
either $e(i,j)\in\mathcal E$ or $e(j,i)\in\mathcal E$, i.e, $\mathcal
E^{s}=\mathcal E^{in}\bigcup\mathcal E^{out}$. Obviously, $R$ is negative
semi-definite. Let $0=\lambda_{1}<\lambda_{2}\le\cdots\le\lambda_{m}$ be
the eigenvalue of $-R$, counting the multiplicities. Let
$\Lambda=diag[\lambda_{1},\cdots,\lambda_{m}]$, from \cite{Rah}, there is a
real orthogonal matrix $\Sigma$ satisfies $-R=\Sigma^{\top}\Lambda\Sigma$.
Thus
\begin{align}\label{RB}
-R=B^{2},
\end{align}
where $B=\Sigma^{\top}\sqrt{\Lambda}\Sigma$. From \cite{Rah}, we know that
$B$ is the unique positive semi-definite matrix satisfies (\ref{RB}).

In the sequel, we need following three matrices: matrix $-R=-(1/2)(\Xi
L+L^{\top}\Xi)$, matrix $Q=\Xi LL^{\top}\Xi$, and matrix
$U=\Xi-\xi\xi^{\top}$ with eigenvalues
$0=\lambda_{1}\le\cdots\le\lambda_{m}$,
$0=\beta_{1}<\beta_{2}\le\cdots\le\beta_{m}$ and
$0=\mu_{1}<\mu_{2}\le\cdots\le\mu_{m}$, (counting their multiplicities),
respectively.

Following three inequalities are used frequently in the rest of the paper:
$$\lambda_{m}x^{T}x\ge\min_{x\bot 1}\{x^{T}(-R)x\}\ge\lambda_{2}x^{T}x,$$
$$\beta_{m}x^{T}x\ge\min_{x\bot 1}\{x^{T}Qx\}\ge\beta_{2}x^{T}x,$$
and
$$\mu_{2}x^{T}x\le\max_{x\bot 1}\{x^{T}Ux\}\le\mu_{m}x^{T}x.$$

Therefore, we have
\begin{align}
\frac{\lambda_{2}}{\beta_{m}}Q\le-R\le\frac{\lambda_{m}}{\beta_{2}}Q,\label{RQ}\\
\frac{\lambda_{2}}{\mu_{m}}U\le-R\le\frac{\lambda_{m}}{\mu_{2}}U.\label{RU}
\end{align}

By \cite{Rwr}, we have
\begin{lemma} \label{lem3} Let $A\in M_{n}(\mathbb R)$ be a be a stochastic matrix. If
$A$ has a spanning tree and positive diagonal elements, then
$\lim_{m\rightarrow\infty}A^m=\mathbf 1v^{\top}$, where $v$ satisfies
$A^{\top}v=v$ and $\mathbf 1^{\top}v=1$. Furthermore, each element of $v$
is nonnegative.
\end{lemma}

\noindent {\bf Notations}: $\|\cdot\|$ represents the Euclidean norm for
vectors or the induced 2-norm for matrices. The notation $\bf 1$ denotes a
column vector with each component 1 and proper dimension. The notation
$\rho(\cdot)$ stands for the spectral radius for matrices and
$\rho_2(\cdot)$ indicates the minimum positive eigenvalue for matrices
which have positive eigenvalues.

\section{Event-triggered control}
Consider following multi-agent system
\begin{eqnarray}
\begin{cases}
\dot{x}_{i}(t)=u_{i}(t)\\
u_{i}(t)=\sum_{j=1}^{m}L_{ij}x_{j}(t_{k(t)}),~i=\onetom\end{cases}\label{mg2}
\end{eqnarray}
where $x_{i}$ represents the state of the agent $i$ at time $t$ and
$L=[L_{ij}]_{i,j=1}^{m}$ is the weighted Laplacian matrix associated with
the underlying graph of the network system, and
$$k(t)=\arg\max_{ k'=1,2,\cdots}\{t_{k'}\le t\}$$
with $t_1=0\le t_2\le t_3\le\cdots$ which are the triggering time points to
be determined.

In order to design the appropriate triggering times, we define the state
measurement error as:
\begin{align}
\Delta x(t)=x(t_k)-x(t),~t\in[t_k,t_{k+1}),~k=0,1,2,\cdots
\end{align}
or
\begin{align}
\Delta x(t)=x(t_{k(t)})-x(t),
\end{align}
where $x(t)=[x_{1}(t),\cdots,x^{m}(t)]^{\top}\in\R^{m}$.

\subsection{Asymmetric and irreducible topology}
In this subsection, we consider the case that $L$ is irreducible.
Equivalently, $\mathcal G$ is strongly connected. To depict the event that
triggers the next coupling term basing time point, we consider the
following candidate Lyapunov function:
\begin{eqnarray}
V(t)=\frac{1}{2}\sum_{i=1}^{m}\xi_{i}(x_{i}(t)-\bar{x}(t))^{2}=\frac{1}{2}
(x(t)-\bar{X}(t))^{\top}\Xi
(x(t)-\bar{X}(t))=\frac{1}{2}x(t)^{\top}Ux(t)\label{V}
\end{eqnarray}
where $\bar{x}(t)=\sum_{i=1}^{m}\xi_{i}x_{i}(t)$ is the weighted average of
$x(t)$ with respect to $\xi$,
$\bar{X}(t)=[\bar{x}(t),\cdots,\bar{x}(t)]^{\top}$. Then, by the definition
of $\Delta x(t)$, $L$, $\dot{\bar{X}}(t)=0$ and (\ref{RQ}), for any $a>0$,
the derivative of $V(t)$ along (\ref{mg2}) is
\begin{align}
&(x(t)-\bar{X}(t))=[1+(t-t_{k}) L](x(t_k)-\bar{X}(t_k))
\end{align}

\begin{align}
&(x(t)-\bar{X}(t))^{\top}\Xi(x(t)-\bar{X}(t))\nonumber\\
=&{(x(t_k)-\bar{X}(t_k))}^{\top}[I+(t-t_{k}) L^{T}]\Xi[1+(t-t_{k}) L](x(t_k)-\bar{X}(t_k))\nonumber\\
=&(x(t_k)-\bar{X}(t_k))^{\top} \Xi(x(t_k)-\bar{X}(t_k))\nonumber\\
&+(t-t_{k})(x(t_k)-\bar{X}(t_k))^{\top} L(x(t_k)-\bar{X}(t_k))\nonumber\\
=&[1+(t-t_{k})](x(t_k)-\bar{X}(t_k))^{\top}R(x(t_k)-\bar{X}(t_k))\nonumber\\
\end{align}

\begin{align}
\frac{d}{dt}V(t)=&(x(t)-\bar{X}(t))^{\top}\Xi(\dot{x}(t)-\dot{\bar{X}}(t))\nonumber\\
=&(x(t)-\bar{X}(t))^{\top}\Xi L(x(t_k)-\bar{X}(t_k))\nonumber\\
=&(x(t_k)-\bar{X}(t_k))^{\top}\Xi L(x(t_k)-\bar{X}(t_k))\nonumber\\
&+(t-t_{k})(x(t_k)-\bar{X}(t_k))^{\top}\Xi L(x(t_k)-\bar{X}(t_k))\nonumber\\
=&[1+(t-t_{k})](x(t_k)-\bar{X}(t_k))^{\top}R(x(t_k)-\bar{X}(t_k))\nonumber\\
\end{align}

\begin{align}&(x(t)-\bar{X}(t))^{\top}\Xi L(x(t)+\Delta x(t))\nonumber\\
=&(x(t)-\bar{X}(t))^{\top}\Xi L(x(t)+\Delta x(t))\nonumber\\
=&x^{\top}(t)\Xi Lx(t)+x^{\top}(t)\Xi L\Delta x(t)\nonumber\\
\le&x^{\top}(t)Rx(t)+\frac{a}{2}x^{\top}(t)\Xi LL^{\top}\Xi x(t)+\frac{1}{2a}\|\Delta x(t)\|^2\nonumber\\
\le&x^{\top}(t)Rx(t)+\frac{a\beta_m}{2\lambda_2}x^{\top}(t)(-R)x(t)+\frac{1}{2a}\|\Delta x(t)\|^2\nonumber\\
=&-(1-\frac{a\beta_m}{2\lambda_2})x^{\top}(t)(-R)x(t)+\frac{1}{2a}\|\Delta x(t)\|^2.\label{dV0.1}
\end{align}
Therefore, we can give
\begin{theorem}\label{thm0.1}
Suppose that $\mathcal G$ is strongly connected. Set $t_{k+1}$ as the time
point such that for some fixed $\gamma\in(0,1)$ and
$0<a<\frac{2\lambda_2}{\beta_m}$
\begin{eqnarray}
t_{k+1}=\max\left\{\tau\ge t_{k}:~\|x(t_k)-x(t)\|\le \sqrt{2\gamma
a(1-\frac{a\beta_m}{2\lambda_2})x^{\top}(t)(-R)x(t)},~\forall
t\in[t_k,\tau]\right\}\label{event0.1}
\end{eqnarray}
Then, system (\ref{mg2}) reaches a consensus; In addition, for all
$i\in\mathcal I$, we have
$$\lim_{t\to\infty}x_{i}(t)=\sum_{j=1}^{m}\xi_{j}x_{j}(0)$$
and
$$\lim_{t\to\infty}x_{i}(t_{k(t)})=\sum_{j=1}^{m}\xi_{j}x_{j}(0).$$
\end{theorem}
\begin{proof}  From (\ref{event0.1}), we know
\begin{align}
\|\Delta x(t)\|=\|x(t_{k(t)})-x(t)\|\le
\sqrt{2\gamma a(1-\frac{a\beta_m}{2\lambda_2})x^{\top}(t)(-R)x(t)}.\label{event0.1tl}
\end{align}
From (\ref{event0.1tl}), (\ref{dV0.1}) and (\ref{RU}), we have
\begin{align}
\frac{d}{dt}V(t)\le&-(1-\frac{a\beta_m}{2\lambda_2})(1-\gamma)x^{\top}(t)(-R)x(t)\nonumber\\
\le&-(1-\frac{a\beta_m}{2\lambda_2})\frac{\lambda_{2}(1-\gamma)}{\mu_{m}}x^{\top}(t)Ux(t)\nonumber\\
=&-(1-\frac{a\beta_m}{2\lambda_2})\frac{2\lambda_{2}(1-\gamma)}{\mu_{m}}V(t)\label{speed}
\end{align}
valid for all $t\ge 0$. It means when $ t_{k+1}\ge t\ge t_{k}$
\begin{align}
V(t)\le V(t_{k})\bigg(exp\bigg\{-(1-\frac{a\beta_m}{2\lambda_2})
\frac{2\lambda_{2}(1-\gamma)}{\mu_{m}} (t-t_{k})\bigg\}\bigg)
\end{align}
{Thus
\begin{align}
V(t)\le V(0)\bigg(exp\bigg\{-(1-\frac{a\beta_m}{2\lambda_2})
\frac{2\lambda_{2}(1-\gamma)}{\mu_{m}} t\bigg\}\bigg)
\end{align}}
This implies that system (\ref{mg2}) reaches a consensus and for all
$i=1,\cdots,m,$
\begin{align}
x_{i}(t)-\sum_{j=1}^{m}\xi_{j}x_{j}(0)=
O\bigg(exp\bigg\{-(1-\frac{a\beta_m}{2\lambda_2})\frac{2\lambda_{2}(1-\gamma)}{\mu_{m}}
{\color{blue} t}\bigg\}\bigg)
\end{align}
and
\begin{align}
x_{i}(t)-x_{i}(t_{k(t)})=
O\bigg(exp\bigg\{-(1-\frac{a\beta_m}{2\lambda_2})\frac{2\lambda_{2}(1-\gamma)}{\mu_{m}}
t\bigg\}\bigg)
\end{align}
The proof is completed.
\end{proof}

Next, we will prove that the above event-triggered rule is realizable, i.e.
the inter-event times $\{t_{k+1}-t_k\}$ are up limited and strictly
positive, which are also known as not exhibiting singular triggering or
Zeno behavior \cite{Joh}. Those are proven in the following theorem.
\begin{theorem}\label{thm0.11}
Under the proposed event-triggered rule (\ref{event0.1}), the next
inter-event time is finite and strictly positive.
\end{theorem}

\begin{proof}
Firstly, we will prove that in case $\|x(t_{k})-\bar{X}(0)\|\neq0$, there
exists a finite triggering time $t_{k+1}>t_k$.

In fact, it is easy to see that for any $t>t_k$, we have

\begin{align*} V(t)\le
V(t_{k})\bigg(exp\bigg\{-(1-\frac{a\beta_m}{2\lambda_2})
\frac{2\lambda_{2}(1-\gamma)}{\mu_{m}} (t-t_{k})\bigg\}\bigg)
\end{align*}
Then,
\begin{align*}
\|x(t)-\bar{X}(0)\|&=\sqrt{\|x(t)-\bar{X}(0)\|^{2}}\le\sqrt{\frac{2}{min_{i}\{\xi_i\}}V(t)}\\
&
\le\sqrt{\frac{2}{min_{i}\{\xi_i\}}V(t_{k})\bigg(exp\bigg\{-(1-\frac{a\beta_m}{2\lambda_2})
\frac{2\lambda_{2}(1-\gamma)}{\mu_{m}} (t-t_{k})\bigg\}\bigg)}
\end{align*}
which means $\|x(t)-\bar{X}(0)\|$ eventually exponentially decreases with
respect to $t$.

On the other hand, it is easy to check that
\begin{align*}
&\sqrt{2\gamma a(1-\frac{a\beta_m}{2\lambda_2})x^{\top}(t)(-R)x(t)}\\
=&\sqrt{2\gamma a(1-\frac{a\beta_m}{2\lambda_2})(x(t)-\bar{X}(0))^{\top}(-R)(x(t)-\bar{X}(0))}\\
\le&\sqrt{2\gamma
a(1-\frac{a\beta_m}{2\lambda_2})\lambda_m}\|x(t)-\bar{X}(0)\|,
\end{align*}
Thus, for sufficient large $t>t_{k}$, we have
\begin{align*}
\|x(t_k)-x(t)\|&\ge\|x(t_{k})-\bar{X}(0)\|-\|x(t)-\bar{X}(0)\|\\
&>\bigg(\sqrt{2\gamma a (1-\frac{a\beta_m}{2\lambda_2})\lambda_m}+1\bigg)\|x(t)-\bar{X}(0)\|-\|x(t)-\bar{X}(0)\|\\
&=\sqrt{2\gamma a(1-\frac{a\beta_m}{2\lambda_2})\lambda_m}~\|x(t)-\bar{X}(0)\|\\
&\ge \sqrt{2\gamma a(1-\frac{a\beta_m}{2\lambda_2})x^{\top}(t)(-R)x(t)}
\end{align*}
On the other hand, $\|x(t_k)-x(t)\|=0$  and $x^{\top}(t)(-R)x(t)\neq0$
at $t=t_{k}$, since $\|x(t_{k})-\bar{X}(0)\|\neq0$. Therefore, 
there exists $t_{k+1} >t_k$ satisfying
\begin{eqnarray}
\|x(t_k)-x(t)\|\le \sqrt{2\gamma
a(1-\frac{a\beta_m}{2\lambda_2})x^{\top}(t)(-R)x(t)}~\forall
t\in[t_k,t_{k+1}]
\end{eqnarray}
This completes the proof that the next inter-event time is finite.

(2) Similar to \cite{Pta}, we can calculate the time derivative of
$\frac{\|\Delta x(t)\|}{\|Bx(t)\|}$.
\begin{align*}
&\frac{d}{dt}\frac{\|\Delta x(t)\|}{\|Bx(t)\|}=\frac{-(\Delta x(t))^{\top}\dot{x}(t)}{\|\Delta x(t)\|\|Bx(t)\|}-\frac{\|\Delta x(t)\|x^{\top}(t)BB\dot{x}(t)}{\|Bx(t)\|^2\|Bx(t)\|}\\
&=\frac{-(\Delta x(t))^{\top}L(x(t)+\Delta x(t))}{\|\Delta x(t)\|\|Bx(t)\|}-\frac{\|\Delta x(t)\|x^{\top}(t)BBL(x(t)+\Delta x(t))}{\|Bx(t)\|^2\|Bx(t)\|}\\
&\le \sqrt{\frac{\rho(L^{\top}L)}{\lambda_2}}+\|L\|\frac{\|\Delta x(t)\|}{\|Bx(t)\|}+\|B\|\sqrt{\frac{\rho(L^{\top}L)}{\lambda_2}}\frac{\|\Delta x(t)\|}{\|Bx(t)\|}+\|L\|\|B\|\frac{\|\Delta x(t)\|^2}{\|Bx(t)\|^2}\\
&=\Big[\|L\|\frac{\|\Delta x(t)\|}{\|Bx(t)\|}+\sqrt{\frac{\rho(L^{\top}L)}{\lambda_2}}\Big]\Big[\|B\|\frac{\|\Delta x(t)\|}{\|Bx(t)\|}+1\Big].
\end{align*}
Via comparison principle, we have $\frac{\|\Delta x(t)\|}{\|Bx(t)\|}\le \phi(t,\phi_0)$, where $\phi(t,\phi_0)$ is the solution of following differential equation
\begin{align*}
\begin{cases}
\frac{d\phi}{dt}=\Big[\|L\|\phi+\sqrt{\frac{\rho(L^{\top}L)}{\lambda_2}}\Big]\Big[\|B\|\phi+1\Big]\\
\phi(t,\phi_0)=\phi_0\end{cases}.
\end{align*}
Hence the inter-event times are bounded from below by the time $\tau_0$ which satisfies $\phi(\tau_0,0)=\sqrt{\gamma2a(1-\frac{a\beta_m}{2\lambda_2})}$. We can calculate $\tau_0$ as follows.
\begin{align*}
\int_{0}^{\sqrt{\gamma2a(1-\frac{a\beta_m}{2\lambda_2})}}\frac{d\phi}{\Big[\|L\|\phi+\frac{\rho(L^{\top}L)}
{\lambda_2}\Big]\Big[\|B\|\phi+1\Big]}=\int_{0}^{\tau_0}dt,
\end{align*}
which yields
\begin{align}\label{tau0}
\tau_0=\begin{cases}
g_1(\|B\|\sqrt{\gamma2a(1-\frac{a\beta_m}{2\lambda_2})})-g_1(0),~if~k_1\ne k_2\\
1-g_2(\sqrt{\gamma2a(1-\frac{a\beta_m}{2\lambda_2})}),~if~k_1=k_2\end{cases}
\end{align}
with $k_1=\|B\|\sqrt{\frac{\rho(L^{\top}L)}{\lambda_2}},~k_2=\|L\|,~g_1(s)=\frac{1}{|k_1-k_2|}ln|\frac{2k_2s+(k_1+k2)-
|k_1-k_2|}{2k_2s+(k_1+k2)+|k_1-k_2|}|,~g_2(s)=\frac{1}{1+s}$.
This completes the proof the next inter-event time is strictly positive.
\end{proof}

Just as the discussion in \cite{Zliu}, we have
\begin{remark}\label{remark1}
Since larger $\tau_0$ implies less control updating times, thus the larger $\tau_0$, the less resources needed for the equipment, like embedded microprocessors, to communicate between agents. On the other hand, the smaller $\dot{V}(t)$ means the faster convergence speed. We know from (\ref{tau0}) that $\tau_0$ is increasing with respect to $\gamma$. Then a larger $\gamma$ leads to less control updating times for each agent, while a smaller $\gamma$ leads to bigger low bound of the system convergence rate according to (\ref{speed}). It should be emphasized that we can not say a smaller $\gamma$ leads to faster system convergence according to (\ref{speed}). Therefore, the protocol designer should choose a proper $\gamma$ so as to make a compromise between the control actuation times and the system convergence speed.
\end{remark}

In the following, we will propose an alternative event-triggered rule.
\begin{corollary}\label{thm0.01}
Suppose that $\mathcal G$ is strongly connected. Set $t_{k+1}$ as the time
point such that for any fixed $\gamma\in(0,1)$
\begin{eqnarray}
t_{k+1}=\max\left\{\tau\ge t_{k}:~|x^{\top}(t)\Xi L(x(t_k)-x(t))|\le
\gamma x^{\top}(t)(-R)x(t),~\forall t\in[t_k,\tau]\right\}.\label{event0.01}
\end{eqnarray}
Then, system (\ref{mg2}) reaches a consensus; In addition, for all
$i\in\mathcal I$, we have
$$\lim_{t\to\infty}x_{i}(t)=\sum_{j=1}^{m}\xi_{j}x_{j}(0)$$ and
$$\lim_{t\to\infty}x_{i}(t_{k(t)})=\sum_{j=1}^{m}\xi_{j}x_{j}(0).$$
\end{corollary}
\begin{proof} From (\ref{event0.01}) and (\ref{dV0.1}), we have
\begin{align*}
\frac{d}{dt}V(t)\le (1-\gamma)x^{\top}(t)Rx(t).
\end{align*}
By the same arguments as in the proof of Theorem \ref{thm0.1}, one can conclude \begin{align*}
V(t)=O\bigg(exp\bigg\{-\frac{2\lambda_{2}(1-\gamma)}{\mu_{m}} t\bigg\}\bigg)
\end{align*}
and
\begin{align}\label{xRx}
x^{\top}(t)(-R)x(t)=O\bigg(exp\bigg\{-\frac{2\lambda_{2}(1-\gamma)}{\mu_{m}} t\bigg\}\bigg).
\end{align}
This implies that system (\ref{mg2}) reaches a consensus and $$\lim_{t\to\infty}x_{i}(t)=\sum_{j=1}^{m}\xi_{j}x_{j}(0).$$ If we can prove that the above event-triggered rule (\ref{event0.01}) is realisable, i.e., the inter-event times are up limited and strictly positive, then we have $$\lim_{t\to\infty}x_{i}(t_{k(t)})=\sum_{j=1}^{m}\xi_{j}x_{j}(0).$$ The proof of the above event-triggered rule (\ref{event0.01}) is realisable can be found in the following theorem. This completes the proof of this theorem.
\end{proof}

\begin{theorem}\label{thm0.01c}
Under the proposed event-triggered rule (\ref{event0.01}), the next inter-event time is finite and strictly positive.
\end{theorem}
\begin{proof} (1) Firstly, we will prove that in case $\|x(t_{k})-\bar{X}(0)\|\neq0$, there
exists a finite triggering time $t_{k+1}>t_k$. Otherwise, we have
\begin{align}\label{assumetion}
t_{k(t)}=t_k,~\forall t\ge t_k.
\end{align}
Thus $$\lim_{t\to\infty}x(t_{k(t)})=x(t_k)\neq X(t_k),$$ and $$\dot{x}(t)=Lx(t_k),~\forall t\ge t_k.$$
Additionally, noting $rank(L)=m-1$, we have $$\lim_{t\to\infty}\dot{x}(t)\neq\bf 0.$$
But from $$\lim_{t\to\infty}x_{i}(t)=\sum_{j=1}^{m}\xi_{j}x_{j}(0)$$ we can conclude $$\lim_{t\to\infty}\dot{x}(t)=\bf 0.$$
This is a contradiction, which means the next triggering time after $t_{k}$, i.e., $t_{k+1}$ exists.
This completes the proof that the next inter-event time is finite.

(2)Let
\begin{align*}
t^{3}_{k+1}=\max\left\{\tau\ge t_{k}:~\|x(t_k)-x(t)\|\le
\frac{\gamma x^{\top}(t)(-R)x(t)}{\|x^{\top}(t)\Xi L\|},~\forall t\in[t_k,\tau]\right\}.
\end{align*}
Obviously, $t_{k+1}$ in (\ref{event0.01}) is not less than above $t^{3}_{k+1}$. Since $0<a<\frac{2\lambda_2}{\beta_m}$, we have
\begin{align}
&\frac{\gamma x^{\top}(t)(-R)x(t)}{\|x^{\top}(t)\Xi L\|}\ge\frac{\gamma x^{\top}(t)(-R)x(t)}{\sqrt{\frac{\beta_m}{\lambda_2}x^{\top}(t)(-R)x(t)}}\nonumber\\
&\ge\gamma \sqrt{\frac{\lambda_2}{\beta_m}x^{\top}(t)(-R)x(t)}
\ge\gamma \sqrt{2a(1-\frac{a\beta_m}{2\lambda_2})x^{\top}(t)(-R)x(t)}.\label{compare}
\end{align}
So $$t_{k+1}-t_{k}\ge t^{3}_{k+1}-t_{k}\ge \tau_0.$$  This completes the proof the next inter-event time is strictly positive.
\end{proof}

\begin{remark}
The next update time $t_{k+1}$ in (\ref{event0.01}) is not less than the next update time in (\ref{event0.1}) if under the same $x(t_k)$. And the influence of $\gamma$ in the above event-triggered rule (\ref{event0.01}) is analogous to those in the discussion of Remark \ref{remark1}.
\end{remark}

\subsection{Asymmetric and reducible topology}
In this subsection, we consider the case that $L$ is reducible. The
following mathematic methods are inspired by the thoughts given in
\cite{Ctp}. By proper permutation, we rewrite $L$ as the following
Perron-Frobenius form:
\begin{eqnarray}
L=\left[\begin{array}{llll}L^{1,1}&L^{1,2}&\cdots&L^{1,K}\\
0&L^{2,2}&\cdots&L^{2,K}\\
\vdots&\vdots&\ddots&\vdots\\
0&0&\cdots&L^{K,K}
\end{array}\right]\label{PF}
\end{eqnarray}
with $L^{k,k}$, with dimension $n_{k}$, associated with the $k$-th strongly
connected component of $\mathcal G$, denoted by $SCC_{k}$, $k=\onetoK$.
Accordingly, define $x^{k}=[x_{1}^{k},\cdots,x_{n_{k}}^{k}]^{\top}$,
corresponding to the $SCC_{k}$. Let $\Delta
x^{k}(t)=x^{k}(t_l)-x^{k}(t),~t\in[t_l,t_{l+1}),~l=0,1,2,\cdots$

If $\mathcal G$ has spanning trees, then each $L^{k,k}$ is irreducible or
has one dimension and for each $k<K$, $L^{k,q}\ne 0$ for at least one
$q>k$. Define an auxiliary matrix
$\tilde{L}^{k,k}=[\tilde{L}^{k,k}_{ij}]_{i,j=1}^{n_{k}}$ as
\begin{eqnarray*}
\tilde{L}^{k,k}_{ij}=\begin{cases}L^{k,k}_{ij}&i\ne j\\
-\sum_{p=1,p\not=i}^{n_{k}}L^{k,k}_{ip}&i=j\end{cases}.
\end{eqnarray*}
Then, let
$D^{k}=L^{k,k}-\tilde{L}^{k,k}=diag[D^{k}_{1},\cdots,D^{k}_{n_{k}}]$, which
is a diagonal semi-negative definite matrix and has at least one diagonal
negative (nonzero). Keep the following property in mind \cite{Wcw}:
\begin{property}
$D^{k}_{i}\ne 0$ if and only if there exists $j\in \bigcup_{l>k}SCC_{l}$
such that there exists an directed link from $j$ to $i+M_{k-1}$, i.e.,
$L^{k,l}_{i,j-M_{l-1}}>0$ for some $j$ and $l>k$.
\end{property}

This implies that for $k=1,\cdots,K$, $L^{k,k}$ is an $M$-matrix. From
\cite{Rah}, one can find some positive definite matrix $\Xi^{k}$ such that
$\Xi^{k}L^{k,k}$ is negative definite. In the following, we are to specify
$\Xi^{k}$. Let ${\xi^{k}}^{\top}$ be the left eigenvector of
$\tilde{L}^{k,k}$ with the eigenvalue zero. Since $\tilde{L}^{k,k}$ is
irreducible, we further specify all components of $\xi^{k}$ positive and
sum equal to 1. Let $\Xi^{k}=diag[\xi^{k}]$. We have
\begin{property}\label{p1}
Under the setup above, $\Xi^{k}L^{k,k}$ is negative definite for all $k<K$.
\end{property}
\begin{proof}
Consider a decomposition of the Euclidean space $\R^{n}$. Define
\begin{align*}
\mathcal S_{n}&=\left\{x\in\R^{n}:\quad x_{i}=x_{j}\quad\forall~i,j=\oneton\right\}\\
\mathcal L_{\zeta}&=\left\{x\in\R^{n}:\quad \sum_{i=1}^{n}\zeta_{i}x_{i}=0\right\}
\end{align*}
for some positive vector $\zeta\in\R^{n}$. In this way, we can decompose
$\R^{n_{k}}=\mathcal S_{n_{k}}\oplus\mathcal L_{\xi^{k}}$. For any
$y\in\R^{n_{k}}$ and $x\ne 0$, we can find a unique decomposition of
$y=y_{S}+y_{L}$ such that $y_{S}\in\mathcal S_{n_{k}}$ and
$y_{L}\in\mathcal L_{\xi^{k}}$. Then, noting
\begin{eqnarray*}
y^{\top}\Xi^{k}L^{k,k}y=y^{\top}\Xi^{k}\tilde{L}^{k,k}y+y^{\top}\Xi^{k}D^{k}y.
\end{eqnarray*}
From Lemma \ref{lem1}, if $y_{L}\ne 0$, then
$y^{\top}\Xi^{k}\tilde{L}^{kk}y<0$; otherwise, $y=y_{S}=\alpha{\bf 1}$ for
some $\alpha\ne 0$, then we have
\begin{eqnarray*}
y^{\top}\Xi^{k}D^{k}y=\sum_{i=1}^{n_{k}}D^{k}_{i}\xi^{k}_{i}\alpha^{2}<0.
\end{eqnarray*}
Therefore, we have $y^{\top}\Xi^{k}L^{k,k}y<0$ in any cases, which implies
that $\Xi^{k}L^{k,k}$ is negative definite. This completes the proof.
\end{proof}

If only consider the $K$-th SCC, we can employ the same rule of event time
sequence $\{t_{l}\}$, as in Theorem \ref{thm0.1}, restricted in this $K$-th
SCC. By Theorem \ref{thm0.1}, the subsystem of (\ref{mg2}) in the $K$-th
SCC can each a consensus with the agreement value equal to
$\nu=\sum_{p=1}^{n_{K}}\xi^{K}_{p}x^{K}_{p}(0)$
and all $x_{j}^{K}(t_{k(t)})$ converge to this value for all
$v_{j+M_{K-1}}\in SCC_{K}$. Here we point out that the nonnegative vector
$[0,\cdots,0,\xi^{K}_{1},\cdots,\xi^{K}_{n_K}]$ is the left eigenvector of
$L$ corresponding to zero. Thus $\nu=\bar{x}(0)$.

Let us consider the $K-1$-th SCC. Construct a candidate Lyapunov function as follows
\begin{align}
V_{K-1}(t)=\frac{1}{2}(x^{K-1}(t)-\nu{\bf 1})^{\top}\Xi^{K-1}(x^{K-1}(t)-\nu{\bf 1}).\label{VK-1r}
\end{align}
Let $R^{K-1}=\frac{1}{2}[\Xi^{K-1}\tilde{L}^{K-1,K-1}+(\Xi^{K-1}\tilde{L}^{K-1,K-1})^{\top}]=[R^{K-1}_{ij}]_{i,j=1}^{n_{K-1}}$, which is all row sums equal to zeros and has zero eigenvalue with algebraic dimension one. Let $Q^{K-1}=\frac{1}{2}[\Xi^{K-1}L^{K-1,K-1}+(\Xi^{K-1}L^{K-1,K-1})^{\top}]=[Q^{K-1}_{ij}]_{i,j=1}^{n_{K-1}}=R^{K-1}+\Xi^{K-1}D^{K-1}$, which is all row sums less than zeros, and from Property \ref{p1}, we know $Q^{K-1}$ is negative definite. Let $\hat{Q}^{K-1}=\Xi^{K-1}L^{K-1,K-1}[\Xi^{K-1}L^{K-1,K-1}]^{\top}$ which is positive (semi-)definite. Similar to (\ref{RQ}), we have
\begin{align}
\hat{Q}^{K-1}\le\frac{\rho(\hat{Q}^{K-1})}{\rho_2(-Q^{K-1})}(-Q^{K-1}),\label{QK-1Qhat}\\
\Xi^{K-1}<I\le\frac{1}{\rho_2(-Q^{K-1})}(-Q^{K-1}).\label{QK-1XiK-1}
\end{align}

The derivative of $V_{K-1}(t)$ along (\ref{mg2}) is
\begin{align}
&\frac{d}{dt}V_{K-1}(t)
=(x^{K-1}(t)-\nu{\bf 1})^{\top}\Xi^{K-1}(
\dot{x}^{K-1})\nonumber\\
=&(x^{K-1}(t)-\nu{\bf 1})^{\top}\Xi^{K-1}\Big\{L^{K-1,K-1}x^{K-1}(t_{k(t)})+L^{K-1,K}x^{K}(t_{k(t)})\Big\}\nonumber\\
=&(x^{K-1}(t)-\nu{\bf 1})^{\top}\Xi^{K-1}\Big\{L^{K-1,K-1}(x^{K-1}(t)-\nu{\bf 1})\nonumber\\
&+L^{K-1,K-1}\Delta x^{K-1}(t)
+L^{K-1,K}(x^{K}(t_{k(t)})-\nu{\bf 1})\Big\}\nonumber\\
=&Q^{K-1}_{3}(t)+(x^{K-1}(t)-\nu{\bf 1})^{\top}\Xi^{K-1}L^{K-1,K-1}\Delta x^{K-1}(t)+Q^{K-1}_{1}(t)
\label{dVK-1r}
\end{align}
where
\begin{align*}
Q^{K-1}_{1}(t)=&(x^{K-1}(t)-\nu{\bf 1})^{\top}\Xi^{K-1}L^{K-1,K}(x^{K}(t_{k(t)})-\nu{\bf 1}),\\
Q^{K-1}_{3}(t)=&(x^{K-1}(t)-\nu{\bf 1})^{\top}\Xi^{K-1}L^{K-1,K-1}(x^{K-1}(t)-\nu{\bf 1})\\
=&[x^{K-1}(t)-\nu{\bf 1}]^{\top}Q^{K-1}[x^{K-1}(t)-\nu{\bf 1}].
\end{align*}
For any $\upsilon^{K-1}>0$, we have
\begin{align*}
Q^{K-1}_{1}(t)\le\upsilon^{K-1} V_{K-1}(t)+F_{1,\upsilon^{K-1}}(t).
\end{align*}
where
\begin{eqnarray*}
F_{1,\upsilon^{K-1}}(t)=\frac{1}{2\upsilon^{K-1}}[L^{K-1,K}(x^{K}(t_{k(t)})-\nu{\bf 1})]^{\top}\Xi^{K-1}L^{K-1,K}(x^{K}(t_{k(t)})-\nu{\bf 1}).
\end{eqnarray*}
According to the discussion of $SCC_{K}$ and Theorem \ref{thm0.1}, for all $p=1,\cdots,n_{K}$, we have
\begin{eqnarray*}
\lim_{t\to\infty}x^{K}_{p}(t_{k(t)})=\nu,
\end{eqnarray*}
exponentially. So,
\begin{align}
\lim_{t\to\infty}F_{1,\upsilon^{K-1}}(t)=0,\label{F1up}
\end{align}
exponentially.

From (\ref{QK-1Qhat}), for any $a^{K-1}>0$, (\ref{dVK-1r}) can be rewritten as
\begin{align}
&\frac{d}{dt}V_{K-1}(t)
=Q^{K-1}_{3}(t)+(x^{K-1}(t)-\nu{\bf 1})^{\top}\Xi^{K-1}L^{K-1,K-1}\Delta x^{K-1}(t)+Q^{K-1}_{1}(t)\nonumber\\
&\le Q^{K-1}_{3}(t)+\frac{1}{2a^{K-1}}\|\Delta x^{K-1}(t)\|^{2}+Q^{K-1}_{1}(t)\nonumber\\
&+\frac{a^{K-1}}{2}(x^{K-1}(t)-\nu{\bf 1})^{\top}\hat{Q}^{K-1}(x^{K-1}(t)-\nu{\bf 1})\nonumber\\
&\le(1-\frac{a^{K-1}\rho(\hat{Q}^{K-1})}{2\rho_2(-Q^{K-1})})Q^{K-1}_{3}(t)+\frac{1}{2a^{K-1}}\|\Delta x^{K-1}(t)\|^{2}+Q^{K-1}_{1}(t).
\label{dVK-1r1}
\end{align}

Analogy, we can define the above quantities for general $k<K$ by replacing $K-1$ by $k$.

Immediately, we have
\begin{theorem}\label{thm0.1r}
Suppose that $\mathcal G$ has spanning tree and $L$ is written in the form of (\ref{PF}). Set $t_{l+1}$ as the time point such that for any fixed $\gamma\in(0,1)$ and $0<a^{k}<\frac{2\rho_2(-Q^{k})}{\rho(\hat{Q}^{k})}$
\begin{align}
t_{l+1}=\min_{k}\big\{\omega^{k}_{l+1}\big\}
\label{event0.1r}
\end{align}
with
\begin{align}\label{tauk}
\omega^{k}_{l+1}=&\max\left\{\tau\ge t_{l}:~\|x^{k}(t_{l})-x^{k}(t)\|
\le\sqrt{2a^{k}\gamma(\frac{a^{k}\rho(\hat{Q}^{k})}{2\rho_2(-Q^{k})}-1)Q^{k}_3(t)},~\forall t\in[t_l,\tau]\right\}.
\end{align}
Then, system (\ref{mg2}) reaches a consensus; In addition,  for all $i\in\mathcal I$, we have $$\lim_{t\to\infty}x_{i}(t)=\sum_{p=1}^{n_{K}}\xi^{K}_{p}x^{K}_{p}(0)$$ and $$\lim_{t\to\infty}x_{i}(t_{k(t)})=\sum_{p=1}^{n_{K}}\xi^{K}_{p}x^{K}_{p}(0).$$
\end{theorem}
\begin{proof} For the $K$-th SCC, the event-triggered rule (\ref{event0.1r}) is the same as (\ref{event0.1}) in Theorem \ref{thm0.1}, since $L$ is written in the form of (\ref{PF}). By Theorem \ref{thm0.1}, we can conclude that under the updating rule (\ref{event0.1r}) for each $v_{j+M_{K-1}}\in SCC_{K}$, the subsystem restricted in $SCC_{K}$ reaches a consensus. And, $\lim_{t\to\infty}x^{K}_{j}(t_{k(t)})=\nu,~j=1,\cdots,n_K$ as well.

In the following, we are to prove that the state of the agent $v_{p+M_{K-2}}\in SCC_{K-1}$ converges to $\nu$ and so it is with $x^{K-1}_{p}(t_{k(t)})$. The remaining can be proved similarly by induction. From (\ref{event0.1r}) and (\ref{dVK-1r1}), we have
\begin{align*}
\frac{d}{dt}V_{K-1}(t)\le(1-\gamma)(1-\frac{a^{K-1}\rho(\hat{Q}^{K-1})}{2\rho_2(-Q^{K-1})})Q^{K-1}_{3}(t)+\upsilon^{K-1} V_{K-1}(t)+F_{1,\upsilon^{K-1}}(t).
\end{align*}
From (\ref{QK-1XiK-1}), we have $$V_{K-1}(t)\le\frac{1}{2\rho_2(-Q^{K-1})}(-Q^{K-1}_{3}(t)).$$ Picking sufficiently small $\upsilon^{K-1}$, there exists some $\upsilon^{K-1}_{0}>0$ such that
\begin{eqnarray*}
\frac{d}{dt}V_{K-1}(t)\le-\upsilon^{K-1}_{0} V_{K-1}(t)+F_{1,\upsilon^{K-1}}(t).
\end{eqnarray*}
From (\ref{F1up}), we have $\lim_{t\to\infty}V_{K-1}(t)=0$ exponentially. This implies that $$\lim_{t\to\infty}x_{p}^{K-1}(t)=\nu$$  exponentially for all $p=1,\cdots,n_{K-1}$. By the same argument in the proof of Theorem \ref{thm0.1}, we can conclude $$\lim_{t\to\infty}x_{p}^{K-1}(t_{k(t)})=\nu$$ exponentially for all $p=1,\cdots,n_{K-1}$, too. Then, we can complete the proof by induction to $SCC_{k}$ for $k<K-1$.
\end{proof}

Like Corollary \ref{thm0.01}, we have
\begin{corollary}\label{thm0.01r}
Suppose that $\mathcal G$ has spanning tree and $L$ is written in the form of (\ref{PF}). Set $t_{l+1}$ as the time point such that for any fixed $\gamma\in(0,1)$ and
\begin{align}
t_{l+1}=\min_{k}\big\{\omega^{k}_{l+1}\big\}
\label{event0.01r}
\end{align}
with
\begin{align}\label{tauk0.01}
\omega^{k}_{l+1}=&\max\bigg\{\tau\ge t_{l}:~|(x^{k}(t)-\nu{\bf 1})^{\top}\Xi^{k}L^{k,k}(x^{k}(t_{l})-x^{k}(t))|
\le-\gamma Q^{k}_3(t),~\forall t\in[t_l,\tau]\bigg\}.
\end{align}
Then, system (\ref{mg2}) reaches a consensus; In addition,  for all $i\in\mathcal I$, we have $$\lim_{t\to\infty}x_{i}(t)=\sum_{p=1}^{n_{K}}\xi^{K}_{p}x^{K}_{p}(0)$$ and $$\lim_{t\to\infty}x_{i}(t_{k(t)})=\sum_{p=1}^{n_{K}}\xi^{K}_{p}x^{K}_{p}(0).$$
\end{corollary}

Similar to (\ref{compare}), we have
\begin{align}
\frac{-\gamma Q^{k}_3(t)}{\|(x^{k}(t)-\nu{\bf 1})^{\top}\Xi^{k}L^{k,k}\|}\ge
\gamma\sqrt{2a^{k}(\frac{a^{k}\rho(\hat{Q}^{k})}{2\rho_2(-Q^{k})}-1)Q^{k}_3(t)}.
\end{align}

\begin{remark}
The next update time $t_{k+1}$ in (\ref{event0.01r}) is not less than the next update time in (\ref{event0.1r}) if under the same $x(t_k)$. As similar with the proof of Theorem \ref{thm0.11} and Theorem \ref{thm0.01c}, by mathematical induction we can prove that the above event-triggered rules (\ref{event0.1r}) and (\ref{event0.01r}) are also realisable, i.e. the inter-event times are up limited and strictly positive. Since the methods are very similar, we omit the proof. And the influence of $\gamma$ in the reducible case is analogous to those in the discussion of Remark \ref{remark1}.
\end{remark}

\begin{remark}
Different SCC could have different $\gamma$ in (\ref{tauk}) and (\ref{tauk0.01}).
\end{remark}

\begin{remark}
When using the above event-triggered rules (\ref{event0.1r}) and (\ref{event0.01r}), in the early stages, one could only consider $SCC_K$ which reaches consensus exponentially, then after some time, one could take $SCC_{K-1}$ into considered which also reaches consensus exponentially. This process goes on each strongly connected component of $\mathcal G$ in a parallel fashion.
\end{remark}

\section{Self-triggered control}
In this section, we present self-triggered strategies for the consensus problem (\ref{mg2}). In the above event-triggered strategy, it is apparently that all agents' states should be observed simultaneously in order to check condition (\ref{event0.1}), (\ref{event0.01}), (\ref{tauk}) or (\ref{tauk0.01}). In the following, the next triggering time $t_{k+1}$ is predetermined at previous triggering time $t_{k}$ or even at the beginning. In time interval $[t_k, t_{k+1})$ ($t_{k+1}$ is waiting to be determined), all agents' states can be formulated as:
\begin{align}
x(t)=(t-t_k)Lx(t_k)+x(t_k).\label{xs}
\end{align}

\subsection{Asymmetric and irreducible topology}
In this subsection, we consider the case of irreducible $L$. Let $\zeta=t-t_{k(t)}$. From (\ref{xs}), we can rewrite all agents' states as
\begin{align*}
x(t)=\zeta Lx(t_{k(t)})+x(t_{k(t)}).
\end{align*}
$\|x(t_{k(t)})-x(t)\|^2=\|Lx(t_{k(t)})\|^2\zeta^{2}$. (\ref{event0.1}) and (\ref{event0.01}) can be written as
\begin{align*}
&x^{\top}(t)Rx(t)=x^{\top}(t_{k(t)})L^{\top}RLx(t_{k(t)})\zeta^{2}+2x^{\top}(t_{k(t)})RLx(t_{k(t)})\zeta
+x^{\top}(t_{k(t)})Rx(t_{k(t)}),\\
&|x^{\top}(t)\Xi L(x(t_{k(t)})-x(t))|^2=|x^{\top}(t_{k(t)})L^{\top}L^{\top}\Xi Lx(t_{k(t)})\zeta^{2}+x^{\top}(t_{k(t)})L^{\top}L^{\top}\Xi x(t_{k(t)})\zeta|^2.
\end{align*}
Denote $\psi=\gamma2a(1-\frac{a\beta_m}{2\lambda_2}),$ solve the following inequality to maximise $\zeta$ so that
\begin{align}
\tau_{l+1}=\max\Big\{&\zeta:\|Lx(t_{k(t)})\|^2s^{2}\le -\psi x^{\top}(t_{k(t)})L^{\top}RLx(t_{k(t)})s^{2}\nonumber\\
&-2\psi x^{\top}(t_{k(t)})RLx(t_{k(t)})s-\psi x^{\top}(t_{k(t)})Rx(t_{k(t)}),~\forall s\in[0,\zeta]\Big\}\label{event0.11}
\end{align}
or
\begin{align}
\tau_{l+1}=\max\Big\{&\zeta:|x^{\top}(t_{k(t)})L^{\top}L^{\top}\Xi Lx(t_{k(t)})s^{2}+x^{\top}(t_{k(t)})L^{\top}L^{\top}\Xi x(t_{k(t)})s|^2\nonumber\\
&\le\gamma^2|x^{\top}(t_{k(t)})L^{\top}RLx(t_{k(t)})s^{2}+2x^{\top}(t_{k(t)})RLx(t_{k(t)})s\nonumber\\
&+x^{\top}(t_{k(t)})Rx(t_{k(t)})|^2,~\forall s\in[0,\zeta]\Big\}.\label{event0.111}
\end{align}
Then, we have the following results
\begin{theorem}\label{thm0.11s}
Suppose that $\mathcal G$ is strongly connected. At each update time $t_{l}$, giving $\tau_{l+1}$ as in (\ref{event0.11}) then the next update time $t_{l+1}=t_{l}+\tau_{l+1}$ with any fixed $\gamma\in(0,1)$ and $0<a<\frac{2\lambda_2}{\beta_m}$.
Then, system (\ref{mg2}) reaches a consensus; in addition, $\lim_{t\to\infty}x_{i}(t)=\sum_{j=1}^{m}\xi_{j}x_{j}(0)$ and $\lim_{t\to\infty}x_{i}(t_{k(t)})=\sum_{j=1}^{m}\xi_{j}x_{j}(0)$ for all $i\in\mathcal I$.
\end{theorem}
\begin{proof} Under the maximisation process (\ref{event0.11}), by the same arguments as in the proof of Theorem \ref{thm0.1}, one can prove this theorem.
\end{proof}

\begin{corollary}\label{thm0.112s}
Suppose that $\mathcal G$ is strongly connected. At each update time $t_{l}$, giving $\tau_{l+1}$ as in (\ref{event0.111}) then the next update time $t_{l+1}=t_{l}+\tau_{l+1}$ with any fixed $\gamma\in(0,1)$.
Then, system (\ref{mg2}) reaches a consensus; in addition, $\lim_{t\to\infty}x_{i}(t)=\sum_{j=1}^{m}\xi_{j}x_{j}(0)$ and $\lim_{t\to\infty}x_{i}(t_{k(t)})=\sum_{j=1}^{m}\xi_{j}x_{j}(0)$ for all $i\in\mathcal I$.
\end{corollary}

Next we will give a simpler self-triggered rule. From (\ref{dV0.1}) and (\ref{xs}), we have
\begin{align}
\frac{d}{dt}V(t)=&(x(t)-\bar{X})^{\top}\Xi\dot{x}(t)\nonumber\\
=&(x(t)-\bar{X})^{\top}\Xi Lx(t_k)\nonumber\\
=&[((t-t_k)L+I)x(t_k)]^{\top}\Xi Lx(t_k)\nonumber\\
=&(t-t_k)x^{\top}(t_k)L^{\top}\Xi Lx(t_k)+x^{\top}(t_k)\Xi Lx(t_k)\nonumber\\
=&(t-t_k)x^{\top}(t_k)L^{\top}\Xi Lx(t_k)+x^{\top}(t_k)Rx(t_k).\label{dV0.12}
\end{align}
Immediately, we have
\begin{theorem}\label{thm0.12}
Suppose that $\mathcal G$ is strongly connected. Set $t_{k+1}$ as the time point such that for any fixed $\gamma\in(0,1)$
\begin{eqnarray}
t_{k+1}\le t_k+\gamma\frac{-x^{\top}(t_k)Rx(t_k)}{x^{\top}(t_k)L^{\top}\Xi Lx(t_k)}.\label{event0.12}
\end{eqnarray}
Then, system (\ref{mg2}) reaches a consensus; In addition,  for all $i\in\mathcal I$, we have $$\lim_{t\to\infty}x_{i}(t)=\sum_{j=1}^{m}\xi_{j}x_{j}(0)$$ and $$\lim_{t\to\infty}x_{i}(t_{k(t)})=\sum_{j=1}^{m}\xi_{j}x_{j}(0).$$
\end{theorem}

\begin{remark}
An very important insight of (\ref{event0.11}) is that it is an one-variable quadratic inequality regarding $s$. And (\ref{event0.111}) is an one-variable inequality with order four. Obviously, (\ref{event0.12}) is simpler than (\ref{event0.11}), and from (\ref{dV0.1}), the next update time $t_{k+1}$ in (\ref{event0.12}) is not less than the next update time in (\ref{event0.11}) if under the same $x(t_k)$.
\end{remark}

By the way, if $\mathcal G$ is symmetric and has a spanning tree, then  zero is an algebraically simple eigenvalue of $L$, $\Xi=\frac{1}{m}I$ and $R=\frac{1}{m}L$. Immediately, from Theorem \ref{thm0.12}, we have
\begin{corollary}\label{coro0.2}
Suppose that $\mathcal G$ is symmetric and has a spanning tree. Set $t_{i+1}$ as the time point such that for any fixed $\sigma\in(0,1)$
\begin{eqnarray}
t_{i+1}\le t_i+\sigma\frac{-x^{\top}(t_i)Lx(t_i)}{x^{\top}(t_i)LLx(t_i)}.\label{event0.12c}
\end{eqnarray}
Then, system (\ref{mg2}) reaches a consensus; in addition, $\lim_{t\to\infty}x_{i}(t)=\sum_{j=1}^{m}\frac{1}{m}x_{j}(0)$ and $\lim_{t\to\infty}x_{i}(t_{k(t)})=\sum_{j=1}^{m}\frac{1}{m}x_{j}(0)$ for all $i\in\mathcal I$.
\end{corollary}
Obviously, the self-triggered rule (\ref{event0.12c}) is simpler than (13) in \cite{Dvd}, which is
\begin{align*}
t_{i+1}-t_i\le\frac{-2\sigma^{2}(Lx(t_i))^{\top}LLx(t_i)+\sqrt{\Delta}}{2(\|Lx(t_i)\|^{2}\|L\|^{2}-\sigma^{2}\|L^2x(t_i)\|^{2})}
\end{align*}
where
\begin{align*}
\Delta=4\sigma^{4}\|(Lx(t_i))^{\top}LLx(t_i)\|^{2}+4\sigma^{4}\|L^2x(t_i)\|^{2}(\|Lx(t_i)\|^{2}\|L\|^{2}-\sigma^{2}\|L^2x(t_i)\|^{2}).
 \end{align*}
Next we will prove that the time interval length given by (\ref{event0.12c}) is bigger than the time interval length given by (13) in \cite{Dvd}.

Let $0=\alpha_1<\alpha_2\le\cdots\alpha_m$ be the eigenvalue of $-L$, counting the multiplicities. Let $\Theta=diag[\alpha_{1},\cdots,\alpha_{m}]$ and $\Gamma$ be the unique real orthogonal matrix satisfies $-L=\Gamma^{\top}\Theta\Gamma$. Let $y(t)=\Gamma^{\top}x(t)=[y_1(t),\cdots,y_m(t)]$. So we have
\begin{align*}
&-x^{\top}(t_i)Lx(t_i)=\sum_{j=2}^{m}\alpha_{j}(y_j(t_i))^2,~x^{\top}(t_i)LLx(t_i)=\sum_{j=2}^{m}(\alpha_{j})^2(y_j(t_i))^2,\\
&-x^{\top}(t_i)LLLx(t_i)=\sum_{j=2}^{m}(\alpha_{j})^3(y_j(t_i))^2.
\end{align*}
Since
\begin{align*}
&\Big(\sum_{j=2}^{m}\alpha_{j}(y_j(t_i))^2\Big)\Big(\sum_{j=2}^{m}(\alpha_{j})^3(y_j(t_i))^2\Big)\\
=&\sum_{j=2}^{m}(\alpha_{j})^4(y_j(t_i))^4+\sum_{j=2}^{m}\sum_{l=2,l\neq j}^{m}\big[\alpha_{j}(\alpha_{l})^3+(\alpha_{j})^3\alpha_{l}\big](y_j(t_i)y_l(t_i))^2\\
\ge&\sum_{j=2}^{m}(\alpha_{j})^4(y_j(t_i))^4+\sum_{j=2}^{m}\sum_{l=2,l\neq j}^{m}2(\alpha_{j}\alpha_{l})^2(y_j(t_i)y_l(t_i))^2\\
=&\Big(\sum_{j=2}^{m}(\alpha_{j})^2(y_j(t_i))^2\Big)\Big(\sum_{j=2}^{m}(\alpha_{j})^2(y_j(t_i))^2\Big),
\end{align*}
then
\begin{align*}
\sigma\frac{-x^{\top}(t_i)Lx(t_i)}{x^{\top}(t_i)LLx(t_i)}\ge\sigma\frac{x^{\top}(t_i)LLx(t_i)}{-x^{\top}(t_i)LLLx(t_i)}=\tau^{1}_i.
\end{align*}
Denote the righthand side of (13) in \cite{Dvd} as $\tau^{0}_i$. Next we will prove $\tau^{1}_i\ge\tau^{0}_i$. From \cite{Dvd}, we know that $\tau^{0}_i$ is the maximum which satisfies (6) in \cite{Dvd}, which is
\begin{align*}
\|e(t)\|\le\sigma\frac{\|Lx(t)\|}{\|L\|},~\forall t\in[t_i,t_i+\tau^{0}_i].
\end{align*}
Since
\begin{align*}
e(t_i+\tau^{1}_i)=x(t_i+\tau^{1}_i)-x(t_i)=\tau^{1}_iLx(t_i),
\end{align*}
then
\begin{align*}
\|e(t_i+\tau^{1}_i)\|&=\|\tau^{1}_iLx(t_i)\|=\sigma\frac{x^{\top}(t_i)LLx(t_i)}{-x^{\top}(t_i)LLLx(t_i)}\|Lx(t_i)\|\\
&\ge\sigma\frac{\|Lx(t_i)\|}{\|L\|}\ge\sigma\frac{\|Lx(t)\|}{\|L\|},~\forall t\in[t_i,t_i+\tau^{0}_i].
\end{align*}
Thus $\tau^{1}_i\ge\tau^{0}_i$. So we can conclude that the time interval length given by (\ref{event0.12c}) is bigger than the time interval length given by (13) in \cite{Dvd}.

At the end of this subsection, under the condition $\mathcal G$ is symmetric and has a spanning tree, we will give a novel self-triggered formulation which not uses agents' states but only relays the system topology. In time interval $[t_k, t_{k+1})$ ($t_{k+1}$ is waiting to be determined), we have:
\begin{align*}
y(t)=(t-t_k)\Theta y(t_k)+y(t_k).
\end{align*}
\begin{corollary}\label{coro0.1}
Suppose that $\mathcal G$ is symmetric and has a spanning tree. Set $\Delta_{k}=t_{k+1}-t_{k}$ as the
inter-event times such that for some fixed $\frac{\alpha_{m}-\alpha_{2}}{\alpha_{m}+\alpha_{2}}\le\gamma<1$
\begin{eqnarray}
\frac{1-\gamma}{\alpha_{2}}\le\Delta_{k}\le\frac{1+\gamma}{\alpha_{m}}.\label{event0.1c}
\end{eqnarray}
Then, system (\ref{mg2}) reaches a consensus; in addition, $\lim_{t\to\infty}x_{i}(t)=\sum_{j=1}^{m}\frac{1}{m}x_{j}(0)$ for all $i\in\mathcal I$.
\end{corollary}
\begin{proof} Let $B_{k}=(t_{k+1}-t_k)\Theta +I,k=0,1,\cdots$. From (\ref{event0.1c}) we know that the absolute value of $B_{k}$'s diagonal elements are all strictly less than 1 except the first diagonal element. Thus
\begin{align*}
y(t)=[(t-t_k)\Theta +I]B_{k-1}\cdots B_{0}y(0),
\end{align*}
and $\lim_{t\to\infty}y(t)=\lim_{k\to\infty}B_{k-1}\cdots B_{0}y(0)=diag[1,0,\cdots,0]y(0)$. So
\begin{align*}
\lim_{t\to\infty}x(t)=\Gamma diag[1,0,\cdots,0]\Gamma^{\top}x(0)=\sum_{j=1}^{m}\frac{1}{m}x_{j}(0) \mathbf 1.
\end{align*}
\end{proof}

\subsection{Asymmetric and reducible topology}
In this subsection, we consider the case of reducible $L$ and we still suppose $L$ is written in the form of (\ref{PF}). Let $\zeta^{p}=t-t_{k(t)}$. From (\ref{xs}), we can rewrite the $SCC_{p}$ agents¡¯ states as:
\begin{align*}
x^{p}(t)=\zeta^{p}\sum^{K}_{j=p} L^{p,j}x^{j}(t_{k(t)})+x^{p}(t_{k(t)}).
\end{align*}
Thus, to specify (\ref{tauk}), we can rewrite
\begin{align*}
&\psi^{p}=2a^{p}\gamma(\frac{a^{p}\rho(\hat{Q}^{p})}{2\rho_2(-Q^{p})}-1),
~\|x^{p}(t_{k(t)})-x^{p}(t)\|^2=\|\sum^{K}_{j=p} L^{p,j}x^{j}(t_{k(t)})\|^2(\zeta^{p})^{2},\\
&Q^{p}_3(t)=[\sum^{K}_{j=p} L^{p,j}x^{j}(t_{k(t)})]^{\top}Q^{p}[\sum^{K}_{j=p} L^{p,j}x^{j}(t_{k(t)})](\zeta^{p})^{2}
+2[\sum^{K}_{j=p} L^{p,j}x^{j}(t_{k(t)})]^{\top}Q^{p}[x^{p}(t_{k(t)})-\nu{\bf 1}]\zeta^{p}\\
&+[x^{p}(t_{k(t)})-\nu{\bf 1}]^{\top}Q^{p}[x^{p}(t_{k(t)})-\nu{\bf 1}]:=\hat{Q}^{p}_3(\zeta^{p}),\\
&|(x^{p}(t)-\nu{\bf 1})^{\top}\Xi^{p}L^{p,p}(x^{p}(t_{k(t)})-x^{p}(t))|=|[\sum^{K}_{j=p} L^{p,j}x^{j}(t_{k(t)})]^{\top}Q^{p}[\sum^{K}_{j=p} L^{p,j}x^{j}(t_{k(t)})](\zeta^{p})^{2}\\
&+[\sum^{K}_{j=p} L^{p,j}x^{j}(t_{k(t)})]^{\top}(L^{p,p})^{\top}\Xi^{p}[x^{p}(t_{k(t)})-\nu{\bf 1}]\zeta^{p}|:=\tilde{Q}^{p}_3(\zeta^{p}).
\end{align*}
Solve the following inequality to maximise $\zeta^{p}$ so that
\begin{align}
\tau^{p}_{l+1}=\max\Big\{\zeta^{p}:\|\sum^{K}_{j=p} L^{p,j}x^{j}(t_{k(t)})\|^2(s)^{2}\le \psi^{p}\hat{Q}^{p}_3(s),~\forall s\in[0,\zeta^{p}]\Big\},\label{event0.21}
\end{align}
or
\begin{align}
\tau^{p}_{l+1}=\max\Big\{\zeta^{p}:\tilde{Q}^{p}_3(s)\le (\gamma \hat{Q}^{p}_3(s))^2,~\forall s\in[0,\zeta^{p}]\Big\}.\label{event0.211}
\end{align}
Then, we have the following results
\begin{theorem}\label{thm0.21s}
Suppose that $\mathcal G$ has spanning tree and $L$ is written in the form of (\ref{PF}). At each update time $t_{l}$, giving $\tau^{1}_{l+1},\cdots,\tau^{K}_{l+1}$ as in (\ref{event0.21}) then the next update time $t_{l+1}=t_{l}+\min_{p}\{\tau^{p}_{l+1}\}$ with any fixed $\gamma\in(0,1)$ and $0<a^{p}<\frac{2\rho_2(-Q^{p})}{\rho(\hat{Q}^{p})}$.
Then, system (\ref{mg2}) reaches a consensus; in addition, $\lim_{t\to\infty}x_{i}(t)=\sum_{j=1}^{n_{K}}\xi^{K}_{j}x^{K}_{j}(0)$ and $\lim_{t\to\infty}x_{i}(t_{k(t)})=\sum_{j=1}^{n_{K}}\xi^{K}_{j}x^{K}_{j}(0)$ for all $i\in\mathcal I$.
\end{theorem}
\begin{proof} Under the maximisation process (\ref{event0.21}), by the same arguments as in the proof of Theorem \ref{thm0.1r}, one can prove this theorem.
\end{proof}

\begin{corollary}\label{thm0.211s}
Suppose that $\mathcal G$ has spanning tree and $L$ is written in the form of (\ref{PF}). At each update time $t_{l}$, giving $\tau^{1}_{l+1},\cdots,\tau^{K}_{l+1}$ as in (\ref{event0.211}) then the next update time $t_{l+1}=t_{l}+\min_{p}\{\tau^{p}_{l+1}\}$ with any fixed $\gamma\in(0,1)$.
Then, system (\ref{mg2}) reaches a consensus; in addition, $\lim_{t\to\infty}x_{i}(t)=\sum_{j=1}^{n_{K}}\xi^{K}_{j}x^{K}_{j}(0)$ and $\lim_{t\to\infty}x_{i}(t_{k(t)})=\sum_{j=1}^{n_{K}}\xi^{K}_{j}x^{K}_{j}(0)$ for all $i\in\mathcal I$.
\end{corollary}

Like Theorem \ref{thm0.12}, next we will give a simpler self-triggered rule. Since $L$ is written in the form of (\ref{PF}), the $SCC_{K-1}$ agents' states can be formulated as:
\begin{align*}
x^{K-1}(t)=(t-t_{k(t)})L^{K-1,K-1}x^{K-1}(t_{k(t)})+(t-t_{k(t)})L^{K-1,K}x^{K}(t_{k(t)})+x^{K-1}(t_{k(t)}).
\end{align*}
Thus
\begin{align}
&\frac{d}{dt}V_{K-1}(t)
=(x^{K-1}(t)-\nu{\bf 1})^{\top}\Xi^{K-1}(
\dot{x}^{K-1}(t))\nonumber\\
=&(x^{K-1}(t)-\nu{\bf 1})^{\top}\Xi^{K-1}\Big\{L^{K-1,K-1}x^{K-1}(t_{k(t)})+L^{K-1,K}x^{K}(t_{k(t)})\Big\}\nonumber\\
=&\Big\{(t-t_{k(t)})L^{K-1,K-1}(x^{K-1}(t_{k(t)})-\nu{\bf 1})+(t-t_{k(t)})L^{K-1,K}(x^{K}(t_{k(t)})-\nu{\bf 1})\nonumber\\
&+(x^{K-1}(t_{k(t)})-\nu{\bf 1})\Big\}^{\top}\Xi^{K-1}\Big\{L^{K-1,K-1}(x^{K-1}(t_{k(t)})-\nu{\bf 1})\nonumber\\
&+L^{K-1,K}(x^{K}(t_{k(t)})-\nu{\bf 1})\Big\}\nonumber\\
=&(t-t_{k(t)})[L^{K-1,K-1}(x^{K-1}(t_{k(t)})-\nu{\bf 1})]^{\top}\Xi^{K-1}L^{K-1,K-1}(x^{K-1}(t_{k(t)})-\nu{\bf 1})\nonumber\\
+&Q^{K-1}_{4}(t)+Q^{K-1}_{5}(t)
\label{dVK-10.2r}
\end{align}
where
\begin{align*}
Q^{K-1}_{4}(t)=&(x^{K-1}(t_{k(t)})-\nu{\bf 1})^{\top}\Xi^{K-1}L^{K-1,K-1}(x^{K-1}(t_{k(t)})-\nu{\bf 1})\\
=&(x^{K-1}(t_{k(t)})-\nu{\bf 1})^{\top}Q^{K-1}(x^{K-1}(t_{k(t)})-\nu{\bf 1}),\\
Q^{K-1}_{5}(t)=&Q^{K-1}_{6}(t)+Q^{K-1}_{7}(t)+Q^{K-1}_{8}(t),\\
Q^{K-1}_{6}(t)=&2(t-t_{k(t)})[L^{K-1,K}(x^{K}(t_{k(t)})-\nu{\bf 1})]^{\top}\Xi^{K-1}L^{K-1,K-1}(x^{K-1}(t_{k(t)})-\nu{\bf 1}),\\
Q^{K-1}_{7}(t)=&[x^{K-1}(t_{k(t)})-\nu{\bf 1}]^{\top}\Xi^{K-1}L^{K-1,K}(x^{K}(t_{k(t)})-\nu{\bf 1}),\\
Q^{K-1}_{8}(t)=&(t-t_{k(t)})[L^{K-1,K}(x^{K}(t_{k(t)})-\nu{\bf 1})]^{\top}\Xi^{K-1}L^{K-1,K}(x^{K}(t_{k(t)})-\nu{\bf 1}).
\end{align*}
From (\ref{QK-1XiK-1}), for any $\upsilon^{K-1}_{5},~\upsilon^{K-1}_{6}>0$, we have
\begin{align*}
Q^{K-1}_{6}(t)\le&\upsilon^{K-1}_{5}(x^{K-1}(t_{k(t)})-\nu{\bf 1})^{\top}(x^{K-1}(t_{k(t)})-\nu{\bf 1})+F_{1,\upsilon^{K-1}_{5}}(t)\\
\le&-\upsilon^{K-1}_{5}\frac{1}{\rho_2(-Q^{K-1})} Q^{K-1}_{4}(t)+F_{1,\upsilon^{K-1}_{5}}(t),\\
Q^{K-1}_{7}(t)\le&\upsilon^{K-1}_{6}(x^{K-1}(t_{k(t)})-\nu{\bf 1})^{\top}(x^{K-1}(t_{k(t)})-\nu{\bf 1})+F_{1,\upsilon^{K-1}_{6}}(t)\\
\le&-\upsilon^{K-1}_{6}\frac{1}{\rho_2(-Q^{K-1})} Q^{K-1}_{4}(t)+F_{1,\upsilon^{K-1}_{6}}(t),
\end{align*}
where
\begin{align*}
&F_{1,\upsilon^{K-1}_{5}}(t)=\frac{1}{4\upsilon^{K-1}_{5}}\|2(t-t_{k(t)})[L^{K-1,K}(x^{K}(t_{k(t)})-\nu{\bf 1})]^{\top}\Xi^{K-1}L^{K-1,K-1}\|^2\\
&F_{1,\upsilon^{K-1}_{6}}(t)=\frac{1}{4\upsilon^{K-1}_{6}}\|\Xi^{K-1}L^{K-1,K}(x^{K}(t_{k(t)})-\nu{\bf 1})\|^2.
\end{align*}
According to the discussion of $SCC_{K}$ and Theorem \ref{thm0.12}, for all $p=1,\cdots,n_{K}$, we have
\begin{eqnarray*}
\lim_{t\to\infty}x^{K}_{p}(t_{k(t)})=\nu,
\end{eqnarray*}
exponentially. So,
\begin{align}
\lim_{t\to\infty}F_{1,\upsilon^{K-1}_{5}}(t)=0,~\lim_{t\to\infty}F_{1,\upsilon^{K-1}_{6}}(t)=0,\lim_{t\to\infty}~Q^{K-1}_{8}(t)=0,\label{F1up2}
\end{align}
exponentially. Immediately, we have

\begin{theorem}\label{thm0.2r}
Suppose that $\mathcal G$ has spanning tree and $L$ is written in the form of (\ref{PF}). Set $t_{l+1}$ as the time point such that for any fixed $\gamma\in(0,1)$
\begin{eqnarray}
t_{l+1}\le t_l+\gamma\min_{p}\{\tau^{p}_{l+1}\}
\end{eqnarray}
with
\begin{eqnarray}
\tau^{p}_{l+1}=\frac{-(x^{p}(t_l)-\nu{\bf 1})^{\top}Q^{p}(x^{p}(t_l)-\nu{\bf 1})}{[L^{p,p}(x^{p}(t_l)-\nu{\bf 1})]^{\top}\Xi^{p}L^{p,p}(x^{p}(t_l)-\nu{\bf 1})}\label{event0.2r}
\end{eqnarray}
Then, system (\ref{mg2}) reaches a consensus; in addition,  for all $i\in\mathcal I$, we have $$\lim_{t\to\infty}x_{i}(t)=\sum_{j=1}^{n_{K}}\xi^{K}_{j}x^{K}_{j}(0)$$ and $$\lim_{t\to\infty}x_{i}(t_{k(t)})=\sum_{j=1}^{n_{K}}\xi^{K}_{j}x^{K}_{j}(0).$$
\end{theorem}
\begin{proof} For the $K$-th SCC, the self-triggered rule (\ref{event0.2r}) is the same as (\ref{event0.12}) in Theorem \ref{thm0.12}, since $L$ is written in the form of (\ref{PF}).
By Theorem \ref{thm0.12}, we can conclude that under the updating rule of $\{t_{l}\}$for each $v_{j+M_{K-1}}\in SCC_{K}$, the subsystem restricted in $SCC_{K}$ reaches a consensus. And, $\lim_{t\to\infty}x^{K}_{j}(t_{k(t)})=\nu$ for all $j=1,\cdots,n_{K}$ as well.

In the following, we are to prove that the state of the agent $v_{p+M_{K-2}}\in SCC_{K-1}$ converges to $\nu$ and so it is with $x^{K-1}_{p}(t_{k(t)})$. The remaining can be proved similarly by induction.

From (\ref{dVK-10.2r}) and (\ref{event0.2r}), we have
\begin{align*}
\frac{d}{dt}V_{K-1}(t)\le&(1-\gamma)Q^{K-1}_{4}(t)-\upsilon^{K-1}_{5}\frac{1}{\rho_2(-Q^{K-1})} Q^{K-1}_{4}(t)+F_{1,\upsilon^{K-1}_{5}}(t)\\
&-\upsilon^{K-1}_{6}\frac{1}{\rho_2(-Q^{K-1})} Q^{K-1}_{4}(t)+F_{1,\upsilon^{K-1}_{6}}(t).
\end{align*}
By the similar argument in the proof of Theorem \ref{thm0.1r}, we can complete the proof.
\end{proof}

\noindent{Finally}, we give a more simple self-triggered formulation which not uses agents' states but only relays the system topology.
\begin{theorem}\label{thm0.3}
Suppose that $\mathcal G$ has a spanning tree, then the self-triggered strategy with a fixed time  interval $T_0 \le\frac{\gamma}{\max_{i}\{-L_{ii}\}}$ for some fixed $\gamma\in(0,1)$ between two continue self-triggered times asymptotically solves the consensus problem (\ref{mg2}), where $L_{11}, L_{22} ,\cdots, L_{mm}$ are the diagonal elements of the Laplacian matrix $L$. In addition, $\lim_{t\to\infty}x_{i}(t)=\sum_{j=1}^{m}\eta_{j}x_{j}(0)$ for all $i\in\mathcal I$, where nonnegative vector $\eta^{\top}=[\eta_{1},\cdots,\eta_{m}]$ is a left eigenvector of $L$ corresponding eigenvalue zero and $\mathbf 1^{\top}\eta=1$.
\end{theorem}
\begin{proof} We have
\begin{align*}
x(t)=L(t-kT_0)x(kT_0)+x(kT_0)=[I+L(t-kT_0)][I+T_0L]^{k}x_{0},~t\in [kT_0,(k+1)T_0).
\end{align*}
Let $A=I+T_0L$, then $A$ satisfies all the conditions demand in Lemma \ref{lem3} under the condition $T_0 \le\frac{\gamma}{\max_{i}\{-L_{ii}\}}$. We have $\lim_{k\rightarrow\infty}A^k=\mathbf 1\eta^{\top}$, where nonnegative vector $\eta$ satisfies $A^{\top}\eta=\eta$ and $\mathbf 1^{\top}\eta=1$. Actually, $\eta$ is a left eigenvector of $L$ corresponding eigenvalue zero. Furthermore
\begin{align*}
 \lim\limits_{t\rightarrow\infty}x(t)=\lim\limits_{k\rightarrow\infty}[I+L(t-kT_0)][I+T_0 L]^{k}x_{0}=\mathbf 1\eta^{\top}x_{0}.
 \end{align*}
This completes the proof.
\end{proof}
\begin{remark}
A similar result could be found in \cite{LC2004,Rjr}, but the condition
required here is weaker since we do not require the graph is strongly
connected but just has a spanning tree.
\end{remark}

\section{Examples}
In this section, two numerical examples are given to demonstrate the effectiveness of the presented results. In order to compare the above principles and the normal continuous control, we write the continuous control here:
\begin{align}\label{mg4}
\dot{x}(t)=Lx(t).
\end{align}

\noindent{\bf Firstly}, consider a network of four agents whose Laplacian matrix is given by
\begin{eqnarray*}
L=\left[\begin{array}{rrrr}-2&2&0&0\\
0&-4&4&0\\
0&3&-7&4\\
4&0&5&-9
\end{array}\right].
\end{eqnarray*}
Obviously, this is an asymmetric strongly connected weighted network described by Figure \ref{fig:1} left. The initial value of each agent is randomly selected within the interval $[-5,5]$ in our simulations. Figure \ref{fig:2} shows the four agents evolve under the triggered principles provided in Theorem \ref{thm0.11s}, Corollary \ref{thm0.112s}, Theorem \ref{thm0.12} and Theorem \ref{thm0.3} with $\gamma=0.9$, and $a=\frac{\lambda_2}{\beta_m}=0.0666$ and initial value $[3.1470,4.0580,-3.7300, 4.1340]^{\top}$, comparing with continuous control, i.e., evolving under (\ref{mg4}).  Under above initial conditions, the consensus value can be computed, $\bar{x}(0)=1.6304$, and $T_0=0.1$ in  Theorem \ref{thm0.3}. The symbol $\cdot$ indicates the agent's triggering times.

\begin{figure}[hbt]
\centering
\includegraphics[width=2.2in]{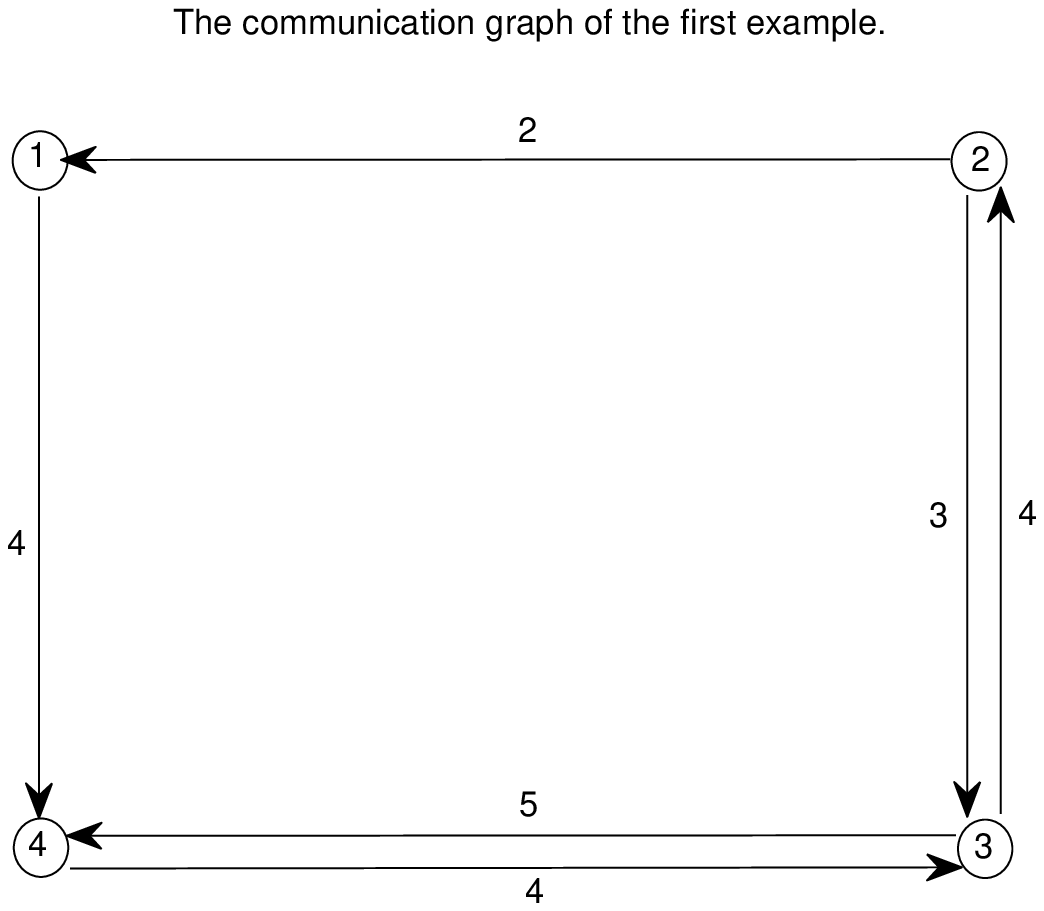}
\includegraphics[width=2.2in]{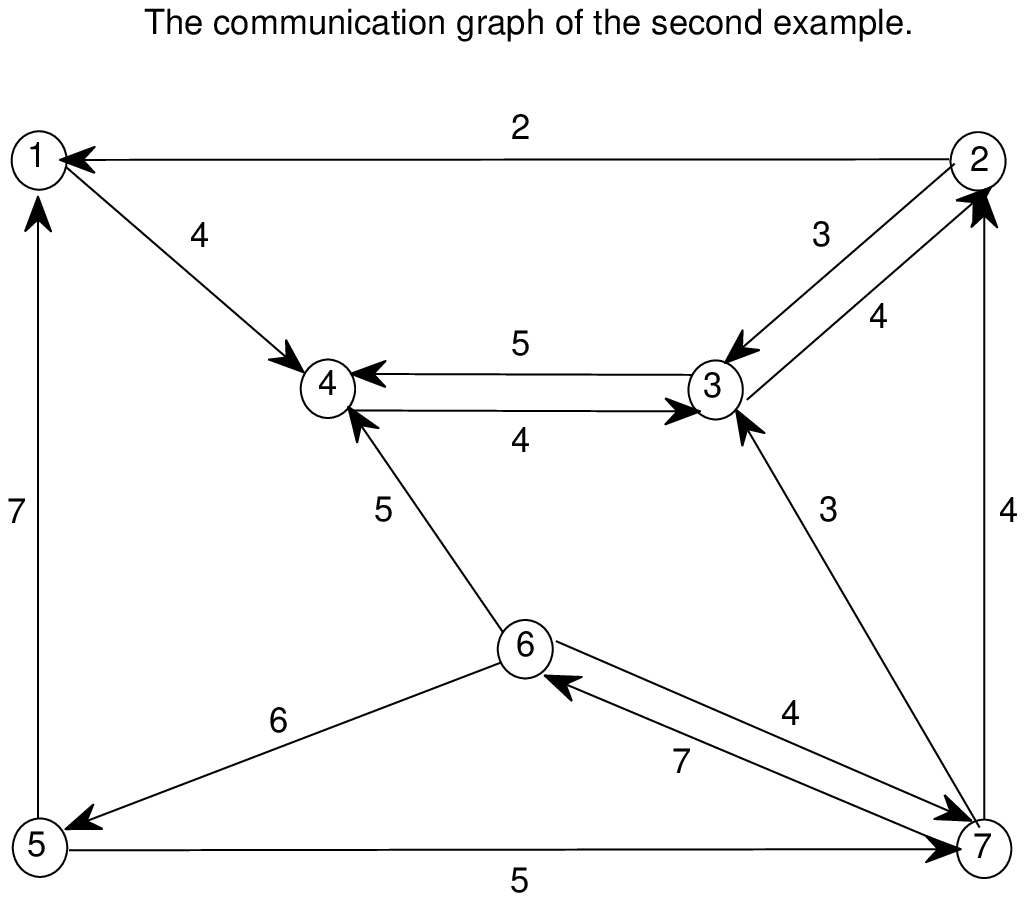}
\caption{The communication graphs.}
\label{fig:1}
\end{figure}

\begin{figure}[hbt]
\centering
\includegraphics[width=2in]{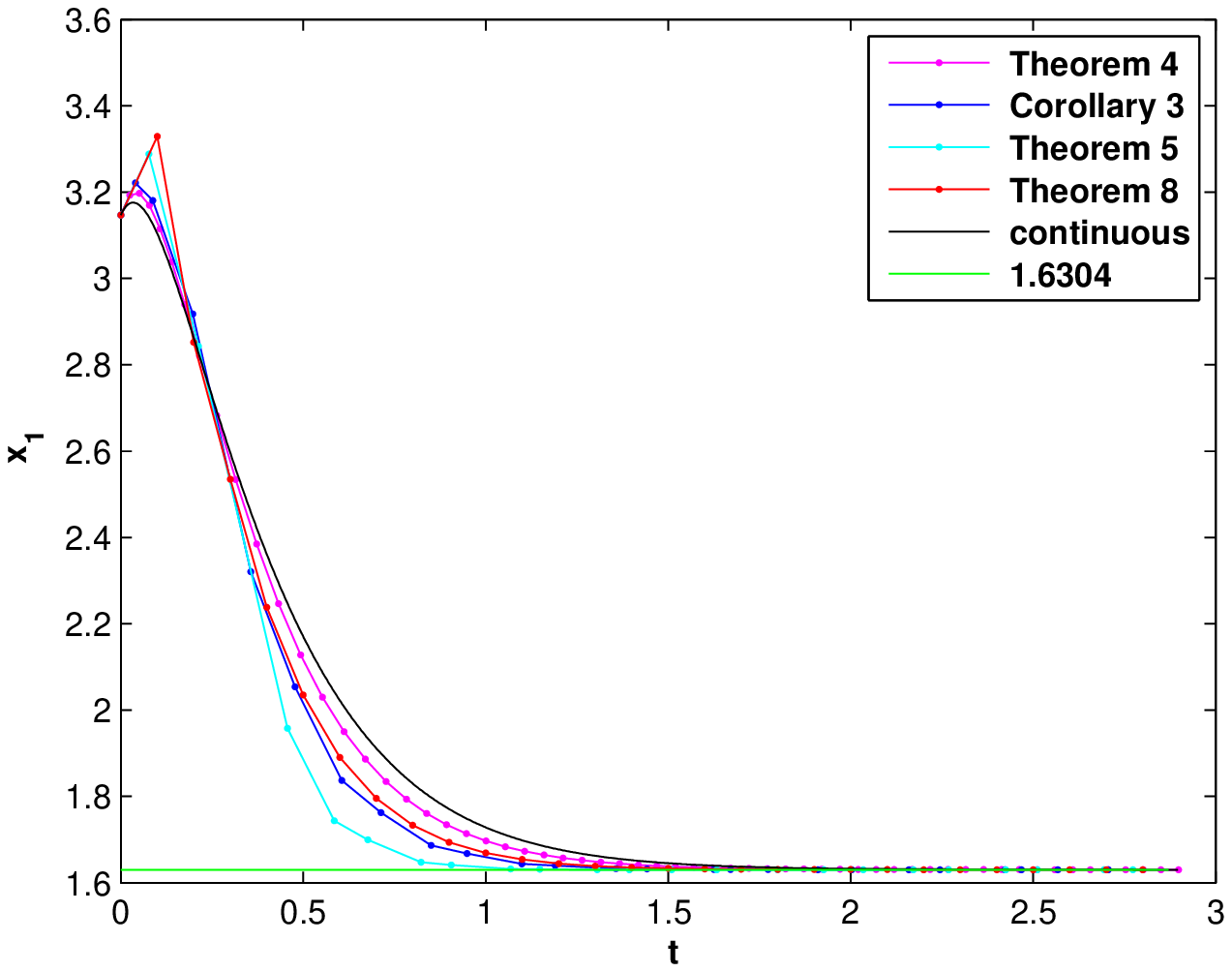}
\includegraphics[width=2in]{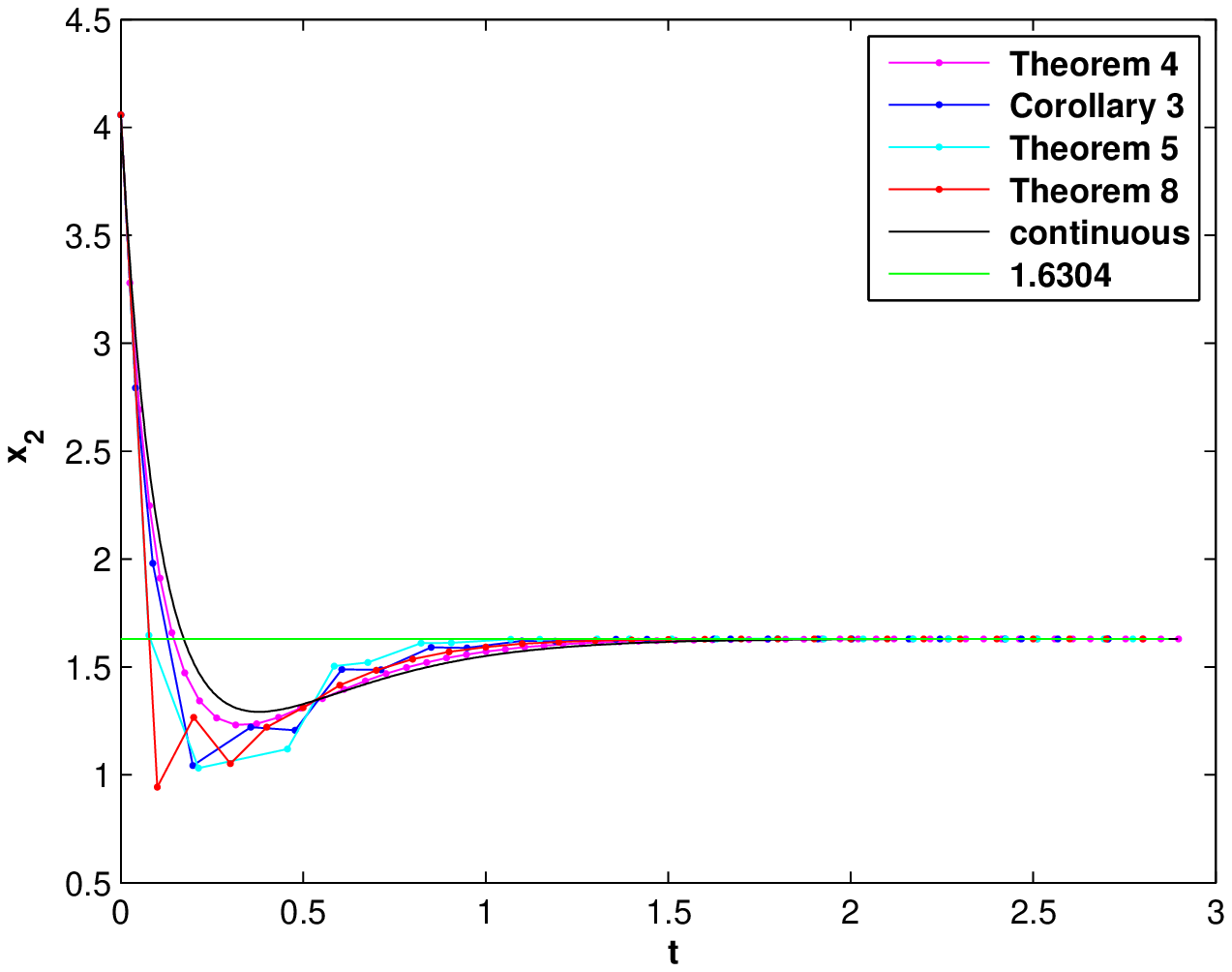}\\
\includegraphics[width=2in]{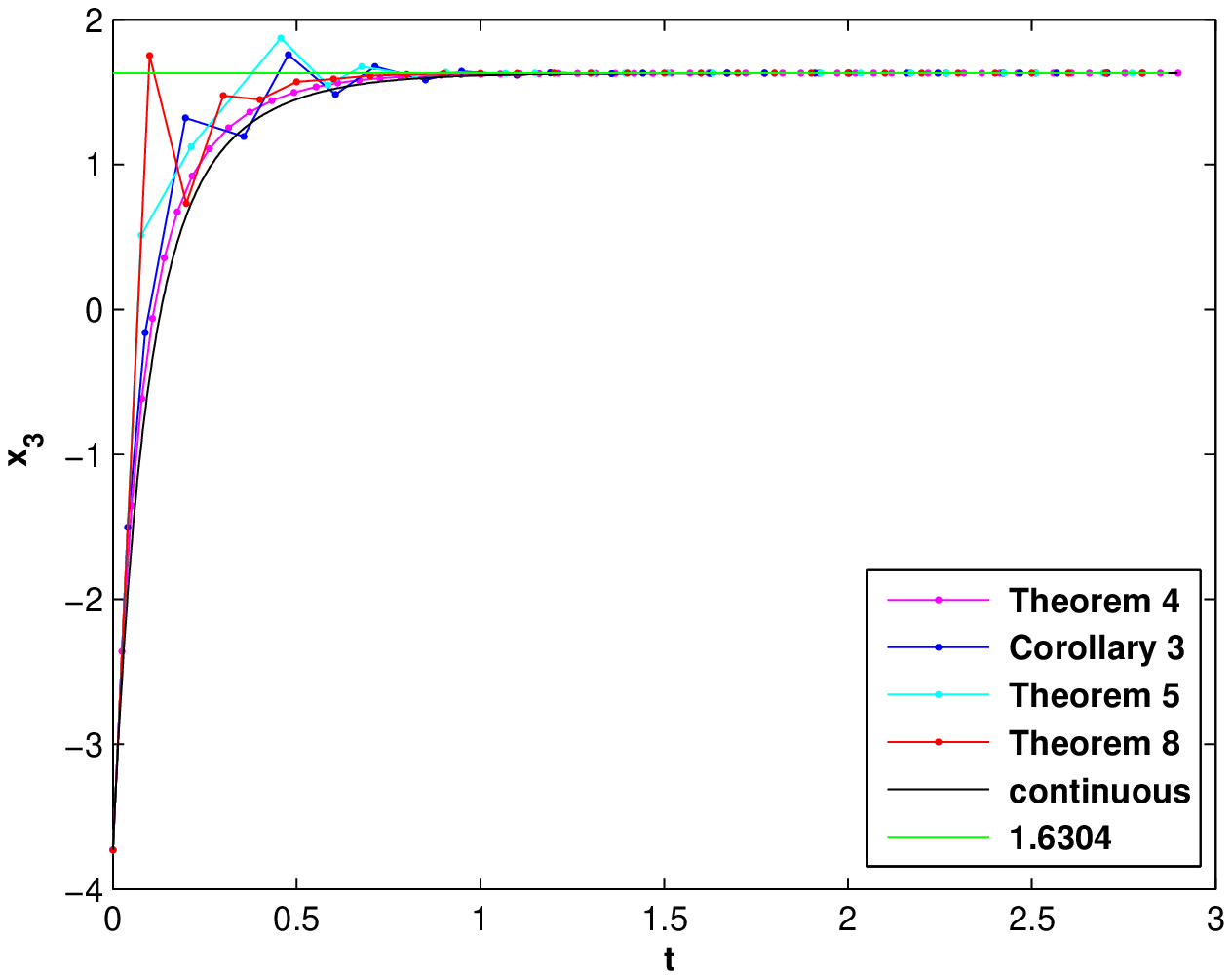}
\includegraphics[width=2in]{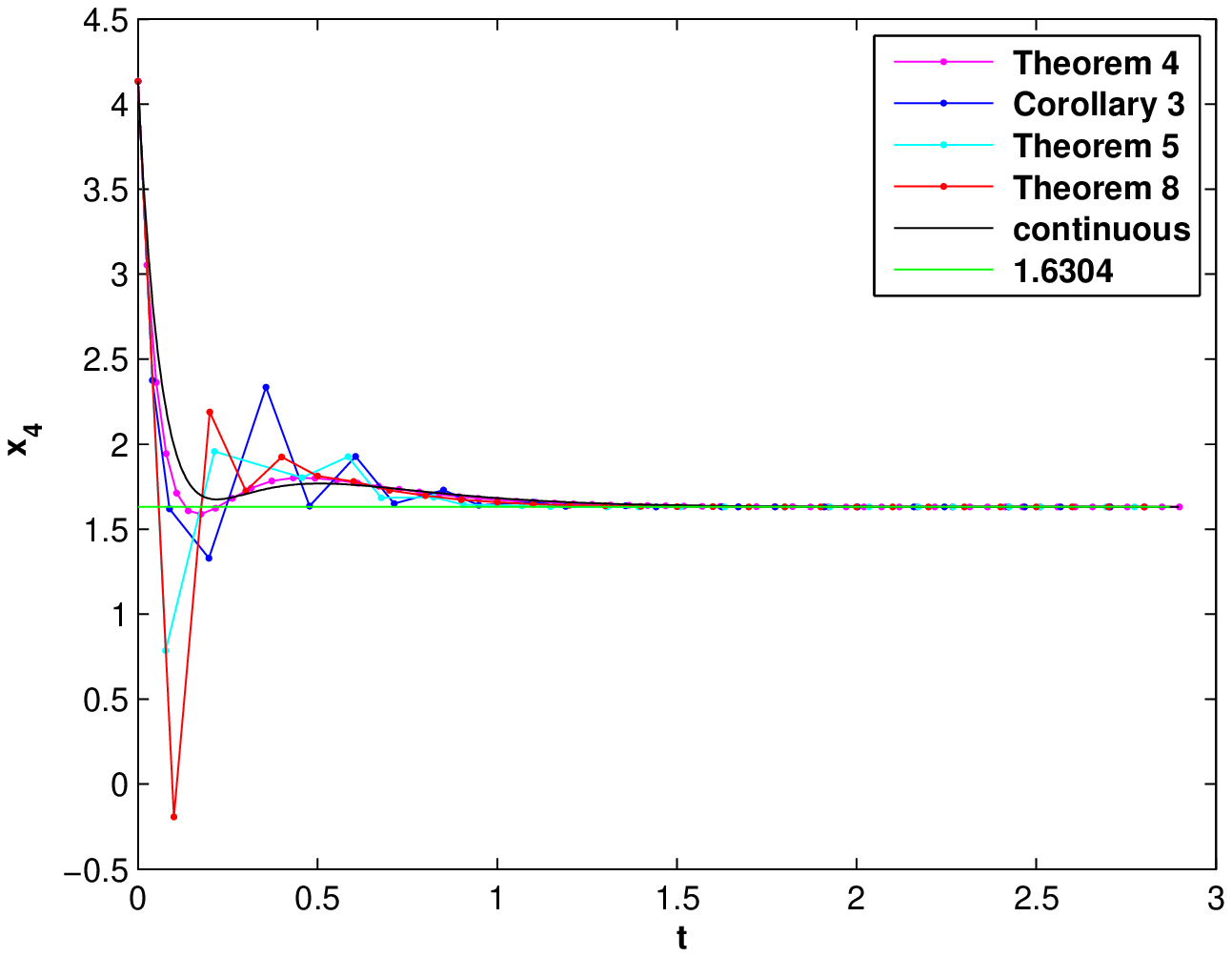}
\caption{Four agents evolve under the event-triggered principles provided in Theorem \ref{thm0.11s}, Corollary \ref{thm0.112s}, Theorem \ref{thm0.12} and Theorem \ref{thm0.3} comparing with continuous control.}
\label{fig:2}
\end{figure}

Then, the parameter $\gamma$ is set to be different values while adopting the triggered principles provided in Theorem \ref{thm0.11s}, Corollary \ref{thm0.112s} and Theorem \ref{thm0.12}. The simulation results are list in Table \ref{tab1}, Table \ref{tab2} and Table \ref{tab3}, respectively. The $T_1$ in the table denotes the first time when $\|x(t)-\bar{X}(0)\|\le 0.0001$, which can be seen as an index representing the convergence speed of the consensus protocol. All the data in this table is the average of 50 runs. It can be seen that all the actual minimum inter-event times are greater than the
corresponding $\tau_0$ calculated by (\ref{tau0}). The minimum value of event interval and the actual number of event decreases with respect to $\gamma$, which is consistent with the theoretical analysis. It is worth noting that $T_1$ also decreases with respect to $\gamma$, which is opposite to usually thinking that $T_1$ increases with respect to $\gamma$. To sum up, the more close to 1 for $\gamma$, the better for the system to realize a consensus.

\begin{table}[htbp]
\begin{tabular}{c|c|c|c|c}
\hline
$\gamma$ & $\tau_0$ calculated by (\ref{tau0}) & the minimum value of event interval & number of event & $T_1$ \\
\hline
0.1  &0.0044&0.0119   &112.88  &2.1021  \\
\hline
0.2  &0.0060&0.0162   &81.12   &2.0713  \\
\hline
0.3  &0.0071&0.0193   &67.16   &2.0524  \\
\hline
0.4  &0.0080&0.0218   &58.78   &2.0350  \\
\hline
0.5  &0.0088&0.0239   &53.16   &2.0239  \\
\hline
0.6  &0.0095&0.0257   &48.82   &2.0062  \\
\hline
0.7  &0.0100&0.0274   &45.54   &1.9944  \\
\hline
0.8  &0.0106&0.0289   &42.84   &1.9813  \\
\hline
0.9  &0.0111&0.0302   &40.72   &1.9748  \\
\hline
\end{tabular}
\caption{Simulation results with different $\gamma$ under the triggered principles provided in Theorem \ref{thm0.11s}.}
\label{tab1}
\end{table}

\begin{table}[htbp]
\begin{tabular}{c|c|c|c|c}
\hline
$\gamma$ & $\tau_0$ calculated by (\ref{tau0}) & the minimum value of event interval & number of event & $T_1$ \\
\hline
0.1  &0.0044&0.0124&107.40&2.5863  \\
\hline
0.2  &0.0060&0.0234&56.36&2.4379  \\
\hline
0.3  &0.0071&0.0333&39.24&2.3251  \\
\hline
0.4  &0.0080&0.0390&31.00&2.2519  \\
\hline
0.5  &0.0088&0.0472&25.56&2.1577  \\
\hline
0.6  &0.0095&0.0478&21.60&2.0541  \\
\hline
0.7  &0.0100&0.0465&18.46&1.9656  \\
\hline
0.8  &0.0106&0.0482&17.06&1.9263  \\
\hline
0.9  &0.0111&0.0483&16.36&1.9058  \\
\hline
\end{tabular}
\caption{Simulation results with different $\gamma$ under the triggered principles provided in Corollary \ref{thm0.112s}.}
\label{tab2}
\end{table}

\begin{table}[htbp]
\begin{tabular}{c|c|c|c|c}
\hline
$\gamma$ & $\tau_0$ calculated by (\ref{tau0}) & the minimum value of event interval & number of event & $T_1$ \\
\hline
0.1  &0.0044&0.0124&110.70&2.6132  \\
\hline
0.2  &0.0060&0.0248&53.02&2.4840  \\
\hline
0.3  &0.0071&0.0371&33.66&2.3436  \\
\hline
0.4  &0.0080&0.0495&23.92&2.1949  \\
\hline
0.5  &0.0088&0.0615&18.52&2.0551  \\
\hline
0.6  &0.0095&0.0719&16.64&1.9798  \\
\hline
0.7  &0.0100&0.0808&16.32&1.9473  \\
\hline
0.8  &0.0106&0.0821&14.82&1.8461  \\
\hline
0.9  &0.0111&0.0787&13.74&1.7856  \\
\hline
\end{tabular}
\caption{Simulation results with different $\gamma$ under the triggered principles provided in Theorem \ref{thm0.12}.}
\label{tab3}
\end{table}

Finally, we compare the triggered principles provided in Theorem \ref{thm0.11s}, Corollary \ref{thm0.112s}, Theorem \ref{thm0.12}, Theorem \ref{thm0.3} with $\gamma=0.9$, and $a=0.0666$ and continuous control. The simulation results are list in Table \ref{tab4}. All the data in this table is the average of 50 runs. It can be seen that  the triggered principle provided in Theorem \ref{thm0.12} is the best, since the corresponding minimum value of event interval is the biggest, the number of event is the smallest and the convergence speed of the consensus protocol is the fastest.

\begin{table}[htbp]
\begin{tabular}{c|c|c|c}
\hline
triggered principles  & the minimum value of event interval & number of event & $T_1$ \\
\hline
Theorem \ref{thm0.11s}  &0.0281   &50.92   &2.4554  \\
\hline
Corollary\ref{thm0.112s} &0.0466   &16.36   &1.8846  \\
\hline
Theorem \ref{thm0.12}   & 0.0774  &13.18   &1.7029  \\
\hline
Theorem \ref{thm0.3}    &/         &21.58   &2.1580  \\
\hline
continuous control      &/         &/        &2.7162  \\
\hline
\end{tabular}
\caption{Simulation results Example 1 with different triggered principles of.}
\label{tab4}
\end{table}

\noindent{\bf Secondly}, we consider a network of seven agents whose Laplacian matrix is given by
\begin{eqnarray*}
L=\left[\begin{array}{rrrrrrr}-9&2&0&0&7&0&0\\
0&-8&4&0&0&0&4\\
0&3&-10&4&0&0&3\\
4&0&5&-14&0&5&0\\
0&0&0&0&-6&6&0\\
0&0&0&0&0&-7&7\\
0&0&0&0&5&4&-9
\end{array}\right].
\end{eqnarray*}
Obviously, this is a asymmetric reducible weighted network with a spanning tree described by Figure \ref{fig:1} right. The seven agents can be divided into two strongly connected components, i.e. the first four agents form a strongly connected component and the rest form anther. The initial value of each agent is also randomly selected within the interval $[-5,5]$ in our simulations. Figure \ref{fig:3} shows the 1st, 3rd, 5th and 7th agents evolve under the triggered principles provided in Theorem \ref{thm0.21s}, Corollary \ref{thm0.211s}, Theorem \ref{thm0.2r} and Theorem \ref{thm0.3} with $\gamma=0.9$, $a^1=\frac{\rho_2(-Q^{1})}{\rho(\hat{Q}^{1})}=0.0580$, $a^2=\frac{\rho_2(-Q^{2})}{\rho(\hat{Q}^{2})}=0.0882$ and initial value $[1.3240,-4.0250,-2.2150,0.4690,4.5750,4.6490,-3.4240]^{\top}$, comparing with continuous control, i.e., evolving under (\ref{mg4}).  Under above initial conditions, the consensus value can be computed, $\nu=2.0409$, and $T_0=0.0643$ in  Theorem \ref{thm0.3}. The symbol $\cdot$ indicates the agent's triggering times.

\begin{figure}[hbt]
\centering
\includegraphics[width=2in]{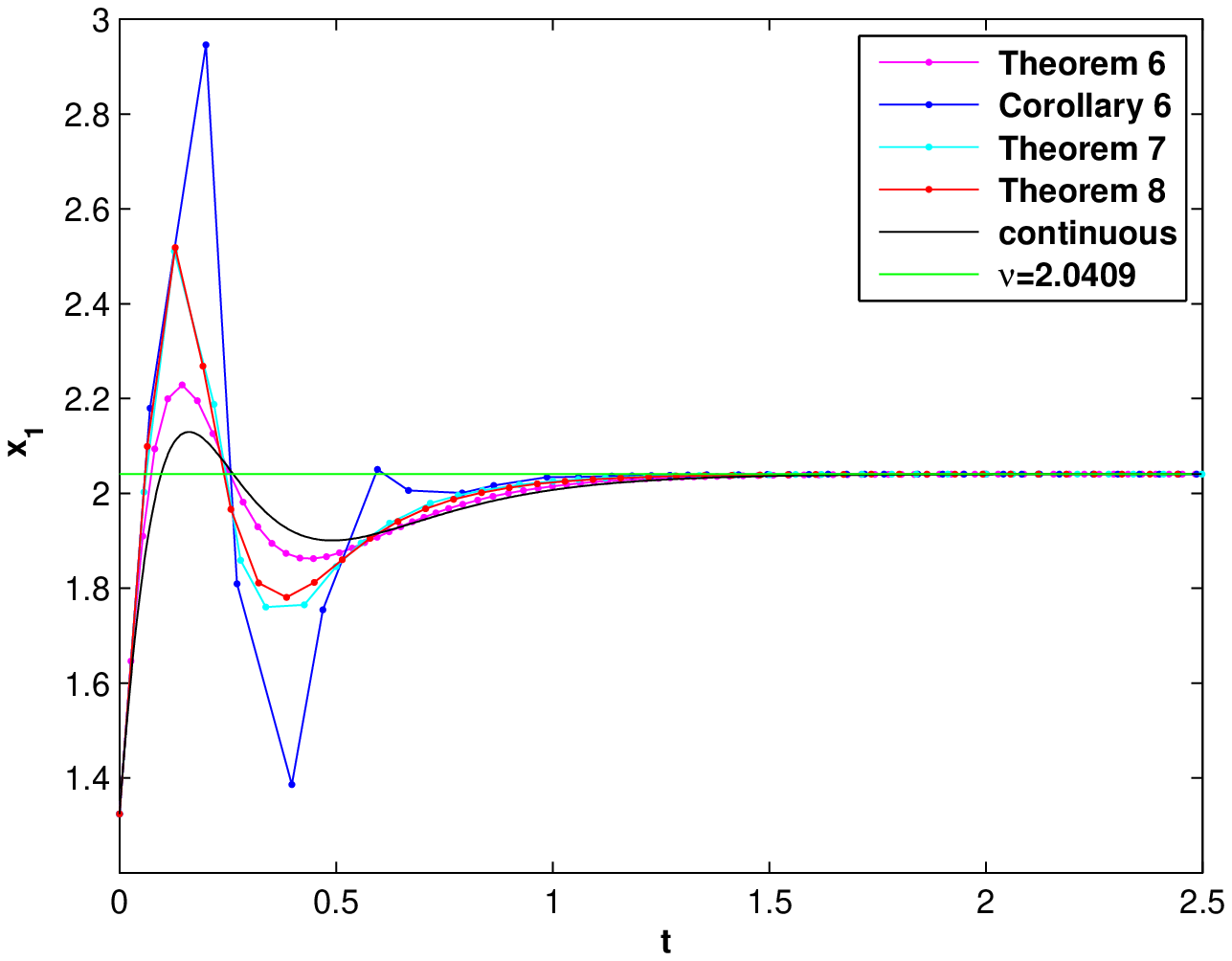}
\includegraphics[width=2in]{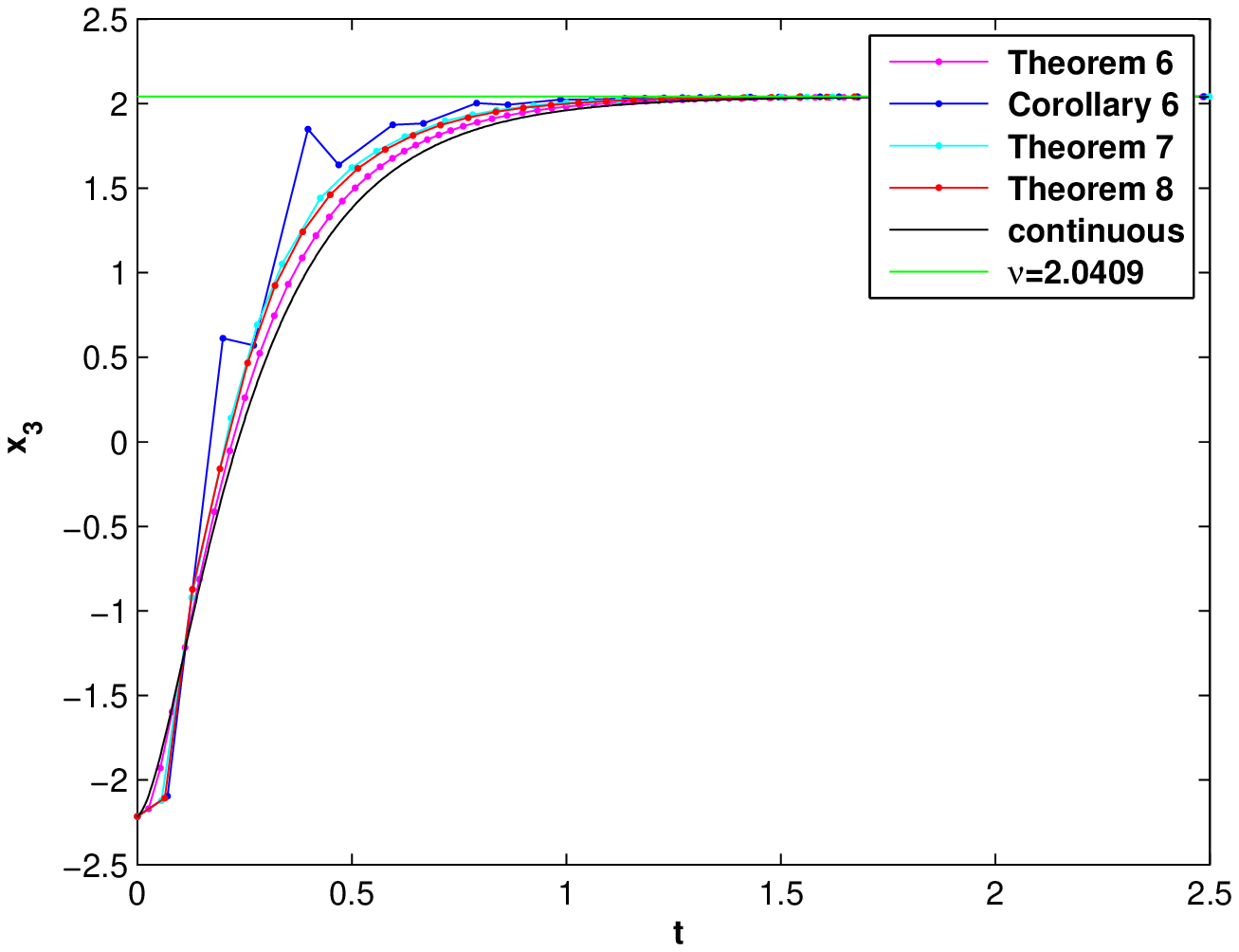}\\
\includegraphics[width=2in]{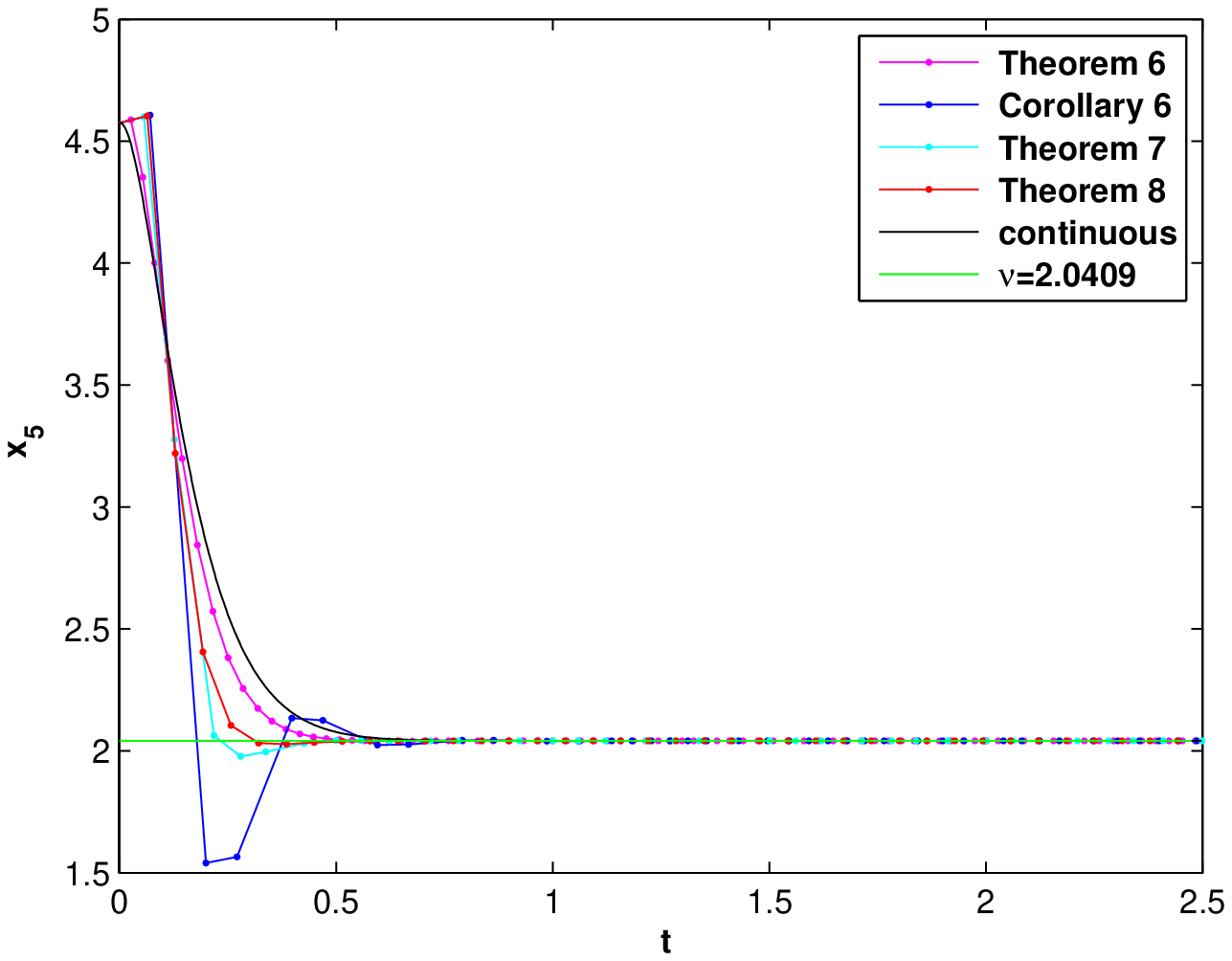}
\includegraphics[width=2in]{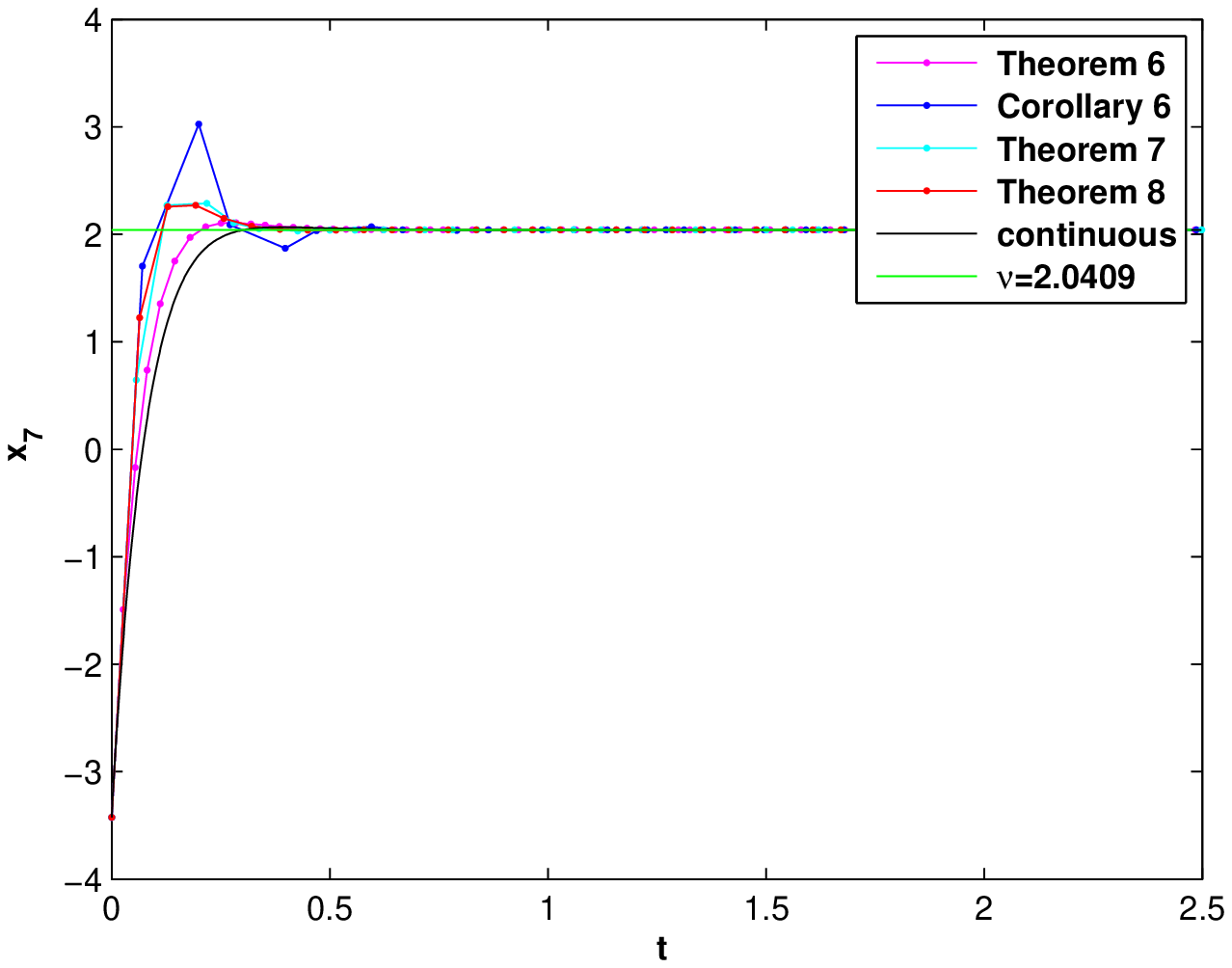}
\caption{The 1st, 3rd, 5th and 7th agents evolve under the event-triggered principles provided in Theorem \ref{thm0.21s}, Corollary \ref{thm0.211s}, Theorem \ref{thm0.2r} and Theorem \ref{thm0.3} comparing with continuous control.}
\label{fig:3}
\end{figure}

Finally, we compare the triggered principles provided in Theorem \ref{thm0.21s}, Corollary \ref{thm0.211s}, Theorem \ref{thm0.2r}, Theorem \ref{thm0.3} with $\gamma=0.9$, $a^1=0.0580$, $a^2=0.0882$, and continuous control. The simulation results are list in Table \ref{tab5}. All the data in this table is the average of 50 runs. It can be seen that the triggered principle provided in Theorem \ref{thm0.2r} is the best, since the corresponding minimum value of event interval is the biggest, the number of event is the smallest and the convergence speed of the consensus protocol is the fastest.

\begin{table}[htbp]
\begin{tabular}{c|c|c|c}
\hline
triggered principles  & the minimum value of event interval & number of event & $T_1$ \\
\hline
Theorem \ref{thm0.21s}  &0.0232   &70.90   &2.3220  \\
\hline
Corollary \ref{thm0.211s} &0.0402   &36.42   &2.1177  \\
\hline
Theorem \ref{thm0.2r}   & 0.0559  &29.54   &2.1063  \\
\hline
Theorem \ref{thm0.3}    &/         &33.52   &2.1549  \\
\hline
continuous control      &/         &/        &2.4688  \\
\hline
\end{tabular}
\caption{Simulation results of Example 2 with different triggered principles.}
\label{tab5}
\end{table}

\section{Conclusion}
In this paper, we first consider centralized event-triggered strategies for multi-agent systems. The triggering times depend on the ratio of a certain measurement error with respect to the norm of a function of the all agents' states. It is proved that if the asymmetric network topology has a spanning tree, then the centralized event-triggered coupling strategy we provide can realize consensus  exponentially for the multi-agent system and singular triggering and Zeno behavior can be both excluded. Then the results are extended to discontinuous monitoring, where each agent computes its next triggering time in advance without having to observe the system¡¯s state continuously and we have pointed out that it is very easy to compute the next triggering time in our principles. In addition, we provide a novel and very simple self-triggered rule (see Theorem \ref{thm0.12} for irreducible case, see Theorem \ref{thm0.2r} for reducible case), and we prove that the time interval length of our rule applied in symmetric topology is bigger comparing with the centralized rule in \cite{Dvd}. Finally, we give a periodic self-triggered strategy. The effectiveness the theoretical results are verified and compared by two examples of numerical simulation. In our numerical simulation, it is worth noting that the time needed to reach consensus decreases with respect to $\gamma$ which is opposite to usually thinking.

In our future paper, inspired by \cite{Pta,Yfg}, we will focus on the distributed event-triggered and self-triggered strategies with push-based feedback and pull-based feedback for multi-agent systems with asymmetric and reducible topologies.

\end{document}